\documentclass{article}
\usepackage{amsmath, amsfonts}
\usepackage{theorem}
\usepackage[dvips]{graphicx}
\usepackage{color}
\usepackage[left=2cm, top=.75cm, bottom=0.35cm,right=2cm]{geometry}

\usepackage{amsmath, amsfonts}

\usepackage{theorem}
\newenvironment{proof}[1][Proof]{\textbf{#1.} }{\ \rule{0.5em}{0.5em}
}
\newtheorem{thm}{Theorem}[section]

\newtheorem{cor}[thm]{Corollary}
\newtheorem{lem}[thm]{Lemma}
\newtheorem{lemma}[thm]{Lemma}
\newtheorem{prop}[thm]{Proposition}

\newtheorem{Assumption}[thm]{Assumption}

{\theorembodyfont{\rmfamily}\newtheorem{rem}{Remark}[section]}
\def\eqd{\,{\buildrel \mathrm{(d)} \over =}\,}

\newcommand{\ba}{\begin{align}}
\newcommand{\ea}{\begin{align}}
\newcommand{\be}{\begin{equation}}

\newcommand{\ee}{\end{equation}}
\newcommand{\bq}{\begin{eqnarray}}
\newcommand{\eq}{\end{eqnarray}}

\newcommand{\half}{\frac{1}{2}}
\newcommand{\bs}{\bigskip}
\newcommand{\nind}{\noindent}
\newcommand{\nn}{\nonumber}

\newcommand{\gm}{\gamma}
\newcommand{\lb}{\lbrace}
\newcommand{\rb}{\rbrace}

\newcommand{\sk}{\smallskip}
\newcommand{\la}{\langle}
\newcommand{\ra}{\rangle}
\newcommand{\ph}{\varphi  }
\newcommand{\e}{\varepsilon}
\newcommand{~}{\,\,\,\,=\,\,\,\,}
\newcommand{\lee}{\,\,\,\,\, \le \,\,\,\,\,}

\newcommand{\baa}{\begin{align}}
\newcommand{\eaa}{\end{align}}

\newcommand{\lm}{\lambda }
\newcommand{\Lm}{\Lambda  }
\newcommand{\Gm}{\Gamma }

\newcommand{\al}{\alpha  }

\newcommand{\ul}{\underline}

\newcommand{\mc}{\mathcal}
\newcommand{\Ex}{\mathbb{E}}
\newcommand{\Pb}{\mathbb{P}}
\newcommand{\Qb}{\mathbb{Q}}
\newcommand{\bc}{\begin{center}}
\newcommand{\ec}{\end{center}}

\begin{document}

\title{\textbf{Small-time, large-time and $H\to 0$ asymptotics \\
for the Rough Heston model}}

\author{Martin Forde
\and
Stefan Gerhold\thanks{TU Wien,
Financial and Actuarial Mathematics,
Wiedner Hauptstra{\ss}e 8/105-1,
A-1040 Vienna, Austria ({\tt
sgerhold@fam.tuwien.ac.at})}
\and
Benjamin Smith\thanks{Dept. Mathematics, King's College London, Strand, London, WC2R 2LS  ({\tt
Benjamin.Smith@kcl.ac.uk})}
}
\maketitle

\begin{abstract}
We characterize the behaviour of the Rough Heston model introduced by Jaisson\&Rosenbaum \cite{JR16} in the small-time, large-time and $\al \to\half$ (i.e. $H\to 0$) limits.  We show that the short-maturity smile scales in qualitatively the same way as a general rough stochastic volatility model (cf.\ \cite{FZ17}, \cite{FGP18a} et al.), and the rate function is equal to the Fenchel-Legendre transform of a simple transformation of the solution to the same Volterra integral equation (VIE) that appears in \cite{ER19}, but with the drift and mean reversion terms removed.  The solution to this VIE satisfies a space-time scaling property which means we only need to solve this equation for the moment values of $p=1$ and $p=-1$ so the rate function can be efficiently computed using an Adams scheme or a power series, and we compute a power series in the log-moneyness variable for the asymptotic implied volatility which yields tractable expressions for the implied vol skew and convexity which is useful for calibration purposes.  We later derive a formal saddlepoint approximation for call options in the \cite{FZ17} large deviations regime which goes to higher order than previous works for rough models.  Our higher order expansion captures the effect of both drift terms, and at leading order is of qualitatively the same form as the higher order expansion for a general model which appears in \cite{FGP18a}.  The limiting asymptotic smile in the large-maturity regime is obtained via a stability analysis of the fixed points of the VIE, and is the same as for the standard Heston model in \cite{FJ11}.  Finally, using L\'{e}vy's convergence theorem, we show that the log stock price $X_t$ tends weakly to a non-symmetric random variable $X^{(\half)}_t$ as $\al \to \half$ (i.e. $H\to 0$) whose mgf is also the solution to the Rough Heston VIE with $\al=\half$, and we show that $X^{(\half)}_t/\sqrt{t}$ tends weakly to a non-symmetric random variable as $t\to 0$, which leads to a non-flat non-symmetric asymptotic smile in the Edgeworth regime, where the log-moneyness  $z=k \sqrt{t}$ as $t\to 0$, and we compute this asymptotic smile numerically.  We also show that the third moment of the log stock price tends to a finite constant as $H\to 0$ (in contrast to the Rough Bergomi model discussed in \cite{FFGS20} where the skew flattens or blows up) and the $V$ process converges on pathspace to a random tempered distribution, which has the same law as the $H=0$ hyper-rough Heston model discussed in Jusselin\&Rosenbaum\cite{JR18} and Jaber\cite{Jab19}.\footnote{We thank Peter Friz, Eduardo Abi Jaber, Martin Larsson and Martin Keller-Ressel for helpful discussions.  S.\ Gerhold acknowledges financial support from the
Austrian Science Fund (FWF) under grant P 30750.}
\end{abstract}

\section{Introduction}


 \quad \cite{JR16} introduced the Rough Heston stochastic volatility model and show that the model arises naturally as the large-time limit of a high frequency market microstructure model driven by two nearly unstable self-exciting Poisson processes (otherwise known as Hawkes process) with a Mittag-Leffler kernel which drives buy and sell orders (a Hawkes process is a generalized Poisson process where the intensity is itself stochastic and depends on the jump history via the kernel).  The microstructure model captures the effects of endogeneity of the market, no-arbitrage, buying/selling asymmetry and the presence of metaorders. \cite{ER19} show that the characteristic function of the log stock price for the Rough Heston model is the solution to a fractional Riccati equation which is non-linear (see also \cite{EFR18} and \cite{ER18}), and the variance curve for the model evolves as $d\xi_u(t)=\kappa(u-t)\sqrt{V_t} dW_t$, where $\kappa(t)$ is the kernel for the $V_t$ process itself multiplied by a \textit{Mittag-Leffler} function (see Proposition \ref{prop:Mittag} below for a proof of this).  Theorem 2.1 in \cite{ER18} shows that a Rough Heston model conditioned on its history up to some time is still a Rough Heston model, but with a time-dependent mean reversion level $\theta(t)$ which depends on the history of the $V$ process.  Using Fr\'{e}chet derivatives, \cite{ER18} also show that one can replicate a call option under the Rough Heston model if we assume the existence a tradeable variance swap, and the same type of analysis can be done for the Rough Bergomi model using the Clark-Ocone formula from Malliavin calculus.
See also \cite{DJR19} who introduce the super Rough Heston model to incorporate the strong Zumbach effect as the limit of a market microstructure model driven by quadratic Hawkes process (this model is no longer affine and thus not amenable to the VIE techniques in this paper).


\sk
\sk
\cite{GK19} consider the more general class of affine forward variance (AFV) models of the form $d\xi_u(t)=\kappa(u-t)\sqrt{V_t} dW_t$  (for which the Rough Heston model is a special case).  They show that AFV models arise naturally as the weak limit of a so-called affine forward intensity (AFI) model, where order flow is driven by two generalized Hawkes-type process with an arbitrary jump size distribution, and we exogenously specify the evolution of the conditional expectation of the intensity at different maturities in the future, akin to a variance curve model.  The weak limit here involves letting the jump size tends to zero as the jump intensity tends to infinity in a certain way, and one can argue that an AFI model is more realistic than the bivariate Hawkes model in \cite{ER19}, since the latter only allows for jumps of a single magnitude (which correspond to buy/sell orders).  Using martingale arguments (which do not require considering a Hawkes process as in the aforementioned El Euch\&Rosenbaum articles) they show that the mgf of the log stock price for the affine variance model satisfies a convolution Riccati equation, or equivalently is a non-linear function of the solution to a VIE.

\sk

\cite{GGP19} use comparison principle arguments for VIEs to show that the moment explosion time for the Rough Heston model is finite if and only if it is finite for the standard Heston model.  \cite{GGP19} also establish upper and lower bounds for the explosion time, and show that the critical moments are finite for all maturities, and formally derive refined tail asymptotics for the Rough Heston model using Laplace's method.  A recent talk by M.Keller-Ressel (joint work with Majid) states an alternate upper bound for the moment explosion time for the Rough Heston model, based on a comparison with a (deterministic) time-change of the standard Heston model, which they claim is usually sharper than the bound in \cite{GGP19}.

\sk
\cite{JP20} compute a small-time LDP on pathspace for a more general class of stochastic Volterra models in the same spirit as the classical Freidlin-Wentzell LDP for small-noise diffusion.  More specifically, for a simple Volterra system of the form
\bq
Y_t &=& Y_0\,+\, \int_0^t K_2(t-s) \zeta(Y_s) dW_s \label{eq:YY} \,
\eq
we have the corresponding deterministic system:
\bq
Y_t &=& Y_0\,+\, \int_0^t K_2(t-s) \zeta(Y_s) v_s ds \nn \,
\eq
where $v \in L^2([0,T])$.  When $K_2(t)=const.t^{H-\half}$ the right term is proporitional to the $\al$-th fractional integral of $\zeta v$ (where $\al=H+\half$), and in this case \cite{JP20} show that $Y_{\e (.)}$ satisfies an LDP as $\e \to 0$ with rate function
\bq
I_Y(\ph) &=&\half const. \times \int_0^T (\frac{D^{\al}(\ph(.)-\ph(0))(t)}{\zeta(\ph(t))})^2 dt \nn
\eq
(see Proposition 4.3 in \cite{JP20}) in terms of the rate function of the underlying Brownian motion which is well known from Schilder's theorem (one can also add drift terms into \eqref{eq:YY} which will not affect $I_Y$).  The corresponding LDP for the log stock price is then obtained using the usual contraction principle method, so the rate function has a variational representation, and does not involve Volterra integral equations. 
\sk

Corollary 7.1 in \cite{FGP18a} provides a sharp small-time expansion in the \cite{FZ17} large deviations regime (valid for $x$-values in some interval) for a general class of Rough Stochastic volatility models using regularity structures, which provides the next order correction to the leading order behaviour obtained in \cite{FZ17}, and some earlier intermediate results in Bayer et al. \cite{BFGHS18}.  
\cite{FSV19} derive formal small-time Edgeworth expansions for the Rough Heston model by solving a nested sequence of linear VIEs.  The Edgeworth-regime implied vol expansions in \cite{EFGR19} and  \cite{FSV19} both include an additional $O(T^{2H})$ term, which itself contains an at-the-money, convexity and higher order correction term, which are important effects to capture for these approximations to be useful in practice.

\sk

 In this article, we establish small-time and large-time large deviation principles for the Rough Heston model, via the solution to a VIE, and we translate these results into asymptotic estimates for call options and implied volatility.  The solution to the VIE satisfies a certain scaling property which means we only have to solve the VIE for the moment values of $p=+1$ and $-1$, rather than solving an entire family of VIEs.  Using the Lagrange inversion theorem, we also compute the first three terms in the power series for the asymptotic implied volatility $\hat{\sigma}(x)$.  We later derive formal asymptotics for the small-time moderate deviations regime and a formal saddlepoint approximation for European call options in the original \cite{FZ17} large deviations regime which goes to higher order than previous works for rough models, and captures the effect of the mean reversion term and the drift of the log stock price, and we discuss practical issues and limitations of this result. Our higher order expansion is of qualitatively the same form as the higher order expansion for a general model in Theorem 6 in \cite{FGP18a} (their expansion is not known to hold for large $x$-values since in their more general setup there are additional complications with focal points, proving non-degeneracy etc.).  For the large time, large log-moneyness regime, we show that the asymptotic smile is the same as for the standard Heston model as in \cite{FJ11}, and we briefly outline how one could go about computing the next order term using a saddlepoint approximation, in the same spirit as \cite{FJM11}.

\sk
\sk

 In the final section, using L\'{e}vy's convergence theorem and result from \cite{GLS90} on the continuous dependence of VIE solutions as a function of a parameter in the VIE, we show that the log stock price $X_t$ (for $t$ fixed) tends weakly as $\al \to \half$  to a random variable $X^{(\half)}_t$ whose mgf is also the solution to the Rough Heston VIE with $\al=\half$ and whose law is non-symmetric when $\rho \ne 0$.  From this we  show that $X^{(\half)}_t/\sqrt{t}$ tends weakly to a non-symmetric random variable as $t\to 0$, which leads to a non-trivial asymptotic smile in the Edgeworth (or central limit theorem) regime.
where the log-moneyness scale as $z=k \sqrt{t}$ as $t\to 0$.  We also show that the third moment of the log stock price for the driftless version of the model tends to a finite constant as $H\to 0$ (in constrast to the Rough Bergomi model discussed in \cite{FFGS20} where the skew flattens or blows up depending on the vol-of-vol parameter $\gm$) and using the expression in \cite{JLP19} for $\Ex(e^{\int_0^T f(T-t) V_t dt })$, we show that $V$ converges to a random tempered distribution whose characteristic functional also satisfies a non-linear VIE
and (from Theorem 2.5 in \cite{Jab19}) this tempered distribution has the same law as the $H=0$ hyper rough Heston model.


\section{Rough Heston and other variance curve models - basic properties}

\sk
In this section, we recall the definition and basic properties and origins of the Rough Heston model, and more general affine and non-affine forward variance models.  Most of the results in this section are given in various locations in \cite{ER18},\cite{ER19} and \cite{GK19}, but for pedagogical purposes we found it instructive to collate them together in one place.

\sk
\sk
Let $(\Omega,\mc{F},\Pb)$ denote a probability space with filtration $(\mc{F}_t)_{t \ge 0}$ which satisfies the usual conditions, and consider the Rough Heston model for a log stock price process $X_t$ introduced in \cite{JR16}:
\begin{align}
dX_t &= -\half V_t dt+\sqrt{V_t} dB_t
\nn \\
V_t &= V_0+\frac{1}{\Gamma(\al)} \int_0^t (t-s)^{\al-1}  \lm(\theta-V_s) ds  +\frac{1}{\Gamma(\al)} \int_0^t (t-s)^{\al-1}  \nu \sqrt{V_s}dW_s \label{eq:Vt}
\end{align}
for $\al \in (\half,1)$, $\theta>0$, $\lm \ge 0$ and $\nu>0$, where $W$, $B$ are two $\mc{F}_t$-Brownian motions with correlation $\rho \in (-1,1)$.  We assume $X_0=0$ and zero interest rate without loss of generality, since the law of $X_t-X_0$ is independent of $X_0$.


\subsection{Computing $\Ex(V_t$)}
\sk
\sk

\begin{prop}
\be
\Ex(V_t) = V_0-(V_0-\theta)\int_0^t f^{\al,\lm}(s) ds \label{eq:Ben}
\ee
where $f^{\al,\lm}(t):=\lambda t^{\alpha-1}E_{\alpha,\alpha}(-\lambda t^\alpha)$, and $E_{\al,\beta}(z):=\sum_{n=0}^{\infty}\frac{z^n}{\Gamma(\al n+\beta)}$ denotes the \textit{Mittag-Leffler function}
\end{prop}
\begin{proof} (see also page 7 in \cite{GK19}), and Proposition 3.1 in \cite{ER18} for an alternate proof).  Let $r(t)=f^{\al,\lm}(t)$.  Taking expectations of \eqref{eq:Vt} and using that the expectation of the stochastic integral term is zero, we see that
\be
\Ex(V_t) = V_0+ \frac{1}{\Gamma(\al)}\int_0^t (t-s)^{\al-1}\lm (\theta-\Ex(V_s))dt \,.\label{eq:EVt}
\ee

\nind Let
$
k(t):=\frac{\lambda t^{\alpha-1}}{\Gamma(\alpha)}$
and
$f(t):=\Ex(V_t)-\theta$.  Then we can re-write \eqref{eq:EVt} as
\be
f(t) =(V_0-\theta)-k*f(t) \label{eq:f(t)} \,.
\ee
where $*$ denotes convolution.  Now define the resolvent $r(t)$ as the unique function which satisfies
$
r=k-k*r \,.
$
Then we claim that
\be
f(t) =(V_0-\theta)-r*(V_0-\theta)\,.\nn
\ee
To verify the claim, we substitute this expression into \eqref{eq:f(t)} to get:
\begin{align*}
(V_0-\theta)-k*[( V_0-\theta)-r*(V_0-\theta)]
&=(V_0-\theta)-(V_0-\theta)*(k-k*r)(t)\nn \\
&=(V_0-\theta)-(V_0-\theta)*r(t)\nn \,
\end{align*}
so $(V_0-\theta)-k*f(t)=(V_0-\theta)-(V_0-\theta)*r(t)=f(t)$,
which is precisely the integral equation we are trying to solve.  Taking Laplace transform of both sides of $k-k*r=r$ we obtain
$
\hat{r} =\hat{k} -\hat{k}\hat{r}
$, which we can re-arrange as
\be
\hat{r}= \frac{\hat{k}}{1+\hat{k}}= \frac{\lambda z^{-\al} }{1 +\lambda z^{-\al}} =
\frac{\lambda  }{z^{\al} +\lambda }\nn
\ee
and the inverse Laplace transform of $\hat{r}$ is $r(t)=\lambda t^{\alpha-1}E_{\alpha,\alpha}(-\lambda t^{\alpha})$.
\end{proof}

\sk

\subsection{Computing $\Ex(V_u|\mc{F}_t$)}

\sk
\sk

Now let $\xi_t(u):=\Ex(V_u|\mc{F}_t)$.  Then $\xi_t(u)$ is an $\mc{F}_t$-martingale, and
\be
 \xi_t(u)
= V_0+ \frac{1}{\Gamma(\al)} \int_0^u (u-s)^{\al-1}  \lm(\theta-\Ex(V_s|\mc{F}_t) ds
+  \frac{1}{\Gamma(\al)}\int_0^t (u-s)^{\al-1}\nu \sqrt{V_s}  dW_s\nn\,.
\ee
If $\lm=0$, we can re-write this expression as
\be
d\xi_t(u) =
\frac{1}{\Gamma(\al)}(u-t)^{\al-1}\sqrt{V_t} dW_t  \nn \,.
\ee
\begin{prop}\label{prop:Mittag} (see \cite{ER19}).
For $\lm > 0$
\be
d\xi_t(u) = \kappa(u-t) \sqrt{V_t} dW_t = \kappa(u-t) \sqrt{\xi_t(t)} dW_t \label{eq:AFV}
\ee
where $\kappa$ is the inverse Laplace transform of $\hat{\kappa}(z)=\frac{\nu z^{-\al}}{1+z^{-\al}}$, which is given explicitly by
\be
\kappa(x) = \nu x^{\al-1} E_{\al,\al}(-\lm x^{\al}) \sim \frac{1}{\Gm(\al)}\nu x^{\al-1}\label{eq:Mitt} \,
\ee
as $x\to 0$ (see also page 6 in \cite{GK19} and page 29 in \cite{ER18}).
\end{prop}
\begin{proof}
See Appendix \ref{section:AppA}.
\end{proof}

\sk
\begin{rem}
Integrating \eqref{eq:AFV} and setting $u=t$ we see that
\bq
V_t =\xi_0(t)+\int_0^t \kappa(t-s) \sqrt{V_s} dW_s \,.
\label{eq:Vtt}
\eq
\end{rem}
\begin{rem}
From \eqref{eq:AFV}, we see that $\xi_t(.)$ is Markov in $\xi_t(.)$.  However $V$ is not Markov in itself.
\end{rem}

\subsection{Evolving the variance curve }

\sk
\sk

We simulate the variance curve at time $t>0$ using
\be
\xi_t(u) =\xi_0(u)+ \int_0^t \kappa(u-s) \sqrt{V_s} dW_s \nn \,
\ee
and substituting the expression for $\xi_0(t)=\Ex(V_t)$ in \eqref{eq:Ben} and the expression for $\kappa(t)$ in Proposition \ref{prop:Mittag} (which are both expressed in terms of the Mittag-Leffler function).


\subsection{The characteristic function of the log stock price}

\sk
\sk

 From Corollary 3.1 in \cite{ER19} (see also Theorem 6 in \cite{GGP19}), we know that
 for all $t \ge 0$
\be
\Ex(e^{p X_t})=e^{V_0 I^{1-\al}f(p,t)+\lm \theta I^1 f(p,t)}\label{eq:SP}
\ee for $p$ in some open interval $I \supset [0,1]$, where $f(p,t)$ satisfies
\be
D^{\al}f(p,t) = \half (p^2-p) +   ( p \,\rho \nu-\lm) f(p,t) +\half \nu^2  f(p,t)^2 \label{eq:FDE0} \,
\ee
with initial condition $f(p,0)=0$, where $I^{\al}f$ denotes the fractional integral operator of order $\al$ (see e.g. page 16 in \cite{ER19} for definition) and $D^{\al}$ denotes the fractional derivative operator of order $\al$ (see page 17 in \cite{ER19} for definition).


\subsection{The generalized time-dependent Rough Heston model and fitting the initial variance curve}

\sk

    If we now replace the constant $\theta$ with a time-dependent function $\theta(t)$, then
\be
\Ex(V_t)
 = V_0+ \frac{1}{\Gamma(\al)}\int_0^t (t-s)^{\al-1}\lm (\theta(s)-\Ex(V_s))dt \nn,
\ee
which we can re-arrange as
\be
\Ex(V_t)-V_0+ \lm I^{\al}\Ex(V_t) = \lm I^{\al}\theta(t)\,\nn
\ee
so to make this generalized model consistent with a given initial variance curve $\Ex(V_t)$, we set
\be
\theta(t) = \frac{1}{\lm}D^{\al}(\Ex(V_t)-V_0+ \lm I^{\al}\Ex(V_t))
=\frac{1}{\lm}D^{\al}(\Ex(V_t)-V_0)+ \Ex(V_t) \nn \,
\ee
(see also Remark 3.2, Theorem 3.2 and Corollary 3.2 in \cite{ER18}).

     \subsection{Other affine and non-affine variance curve models}

     \sk

Another well known (and non-affine) variance curve model is the \textit{Rough Bergomi} model, for which
$
d\xi_t(u) =\eta (u-t)^{H-\half}\xi_t(u) dW_t \nn
$
or the standard Bergomi model (with mean reversion) for which
$
d\xi_t(u) = \eta  e^{-\lambda(u-t)}\xi_t(u)  dW_t \nn \,.
$

\section{Small-time asymptotics}

\sk
\subsection{Scaling relations}
Let
\be
d\tilde{X}^{\e}_t = \sqrt{\e} \sqrt{V^{\e}_t} dB_t\label{eq:Xtilde}
\ee
which satisfies:
\be
\tilde{X}^{\e}_t \eqd \tilde{X}_{\e t}\nn
\ee
Then the characteristic function of $\tilde{X}_t$ for $\e=1$ is:
\be \Ex(e^{p \tilde{X}_t}) = e^{V_0I^{1-\al}\psi(p,t)}
\label{eq:CFF}
\ee
where $\psi(p,t)$ satisfies:
\be
D^{\al}\psi(p,t) = \half  p^2+   p \rho \nu   \psi(p,t) +\half \nu^2  \psi(p,t)^2\,
\label{eq:RiccEq}
\ee
with $\psi(p,0)=0$.
We first recall that
$
D^{\al} \psi(p,t) = \frac{d}{dt} \frac{1}{\Gamma(1-\al)} \int_0^t \psi(p,s) (t-s)^{-\al} ds\nn \,.
$
Then
\begin{align}
D^{\al} \psi(p,\e t):= (D^{\al} \psi)(p,\e t) &= \frac{1}{\e}\frac{d}{dt} \frac{1}{\Gamma(1-\al)} \int_0^{\e t} \psi(p,s) (\e t-s)^{-\al} ds\nn \\
&= \frac{1}{\e}\frac{d}{dt} \frac{1}{\Gamma(1-\al)} \int_0^t \psi(p,\e u) (\e t-\e u)^{-\al} \e du\nn \\
&= \e^{-\al}\frac{d}{dt} \frac{1}{\Gamma(1-\al)} \int_0^t \psi(p,\e u) (t-u)^{-\al} du\nn \\
&= \e^{-\al} D^{\al} \psi(p,\e(.))(t)\nn\,.
\end{align}
Combining this with \eqref{eq:RiccEq} we see that
\be
\e^{-\al}D^{\al}(\psi(p,\e .))(t)= \half  p^2+   p \rho \nu   \psi(p,\e t) +\half \nu^2  \psi(p,\e t)^2\,.
\ee
Setting $p\rightarrow \e^{\gm}q$ and multiplying by $\e^{-2 \gm}$ we have
\be
\e^{-\al-2\gm}D^{\al}(\psi( \e^{\gm}q,\e (.)))(t)=
\half  q^2+   q \rho \nu   \e^{-\gm}\psi(\e^{\gm}q,\e t) +\half \nu^2  \e^{-2\lambda}\psi(\e^{\gm}q,\e t)^2
\ee
Now setting $\gm=-\al$ we see that 
\be
D^{\al}(\e^{\al} \psi(\e^{-\al}q,\e (.)))(t) ~
\half  q^2+   q \rho \nu   \e^{\al}\psi(\e^{-\al}q,\e t) +\half \nu^2  \e^{2\al}\psi(\e^{-\al}q,\e t)^2
\ee
with $\psi(\e^{-\al}q,0)=0$.  Thus, we see that $\e^{\al}\psi(\e^{-\al}p,\e t)$ and $\psi(p,t)$ satisfy the same VIE with the same boundary condition, so
\be
\psi(p,t)=\e^{\al}\psi(\e^{-\al}p,\e t) \label{eq:sss}
\ee
From the form of the characteristic function in \eqref{eq:CFF}, the function $\Lambda(p,t):=I^{1-\al}\psi(p,t)$ is clearly of interest too. Using the scaling relation on $\psi(p,t)$:
\bq
I^{1-\al}\psi(p,\e t)&=&\frac{1}{\Gamma(1-\al)}\int_0^{\e t}(\e t-s)^{-\al}\psi(p,s)ds\\
&=&\frac{\e}{\Gamma(1-\al)}\int_0^{ t}(\e t- \e u)^{-\al}\psi(p,\e u) du\\
&=&\frac{\e^{1-\al}}{\Gamma(1-\al)}\int_0^{ t}(t- u)^{-\al}\e^{-\al}\psi(\e^{\al}p,u)du ~ \e^{-2H}I^{1-\al}\psi(\e^{\al}p,t)
\eq
Thus we have established the following lemma:
\begin{lem}
\bq
\Lambda(p,\e t)=\e^{-2H}\Lambda(\e^{\al}p,t)
\eq
in particular
\bq
\Lambda(p,t)=t^{-2H}\Lambda(pt^{\al},1)\,.
\label{eq:SR}
\eq
\end{lem}
\subsection{The small-time LDP}

\sk
\sk


To simplify calculations, we make the following assumption throughout this section:

\begin{Assumption}
$\lm=0$.
\end{Assumption}
\begin{rem}
The formal higher order Laplace asymptotics in subsection \ref{subsection:HigherOrder} indicate that $\lm$ will not affect the leading order small-time asymptotics, i.e. $\lm$ will not affect the rate function, as we would expect from previous works on small-time asymptotics for rough stochastic volatility models.  The assumption that  $\lm=0$ is relaxed in the next section where we consider large-time asymptotics.
\end{rem}

\sk

We now state the main small-time result in the article (recall that $\al=H+\half$):
\begin{thm}\label{thm:small time}
For the Rough Heston model defined in \eqref{eq:Vt}, we have
\be
\lim_{t \to 0}t^{2H}\log \Ex(e^{\frac{p}{t^{\al}} X_t})
= \lim_{t \to 0} t^{2H}\log \Ex(e^{\frac{p}{t^{2H}} \frac{X_t}{t^{\half-H}}}) =\left\{
            \begin{array}{ll}
      \bar{\Lm}(p)   \quad \text{if}\quad  T^*(p)>1\\
+\infty \quad \text{if}\quad\,\,   T^*(p) \le 1
\end{array}\
\right.
\label{eq:Maine}
\ee
where $\bar{\Lm}(p):=V_0\Lm(p)$, $\Lm(p):=\Lm(p,1)$, $\Lambda(p,t):=I^{1-\al}\psi(p,t)$ and $\psi(p,t)$ satisfies
the Volterra differential equation
\be
D^{\al}\psi(p,t) = \half  p^2+   p \rho \nu   \psi(p,t) +\half \nu^2  \psi(p,t)^2\label{eq:VIEE}  \,
\ee
with initial condition $\psi(p,0)=0$, where $T^*(p)>0$ is the explosion time for $\psi(p,t)$ which is finite for all $p \ne 0$ (assuming $\nu>0$).  Moreover, the scaling relation in the previous section show that $\Lm(p)=|p|^{\frac{2H}{\al}}\Lm(\mathrm{sgn}(p),|p|^{\frac{1}{\al}})$, so in fact we only need to solve \eqref{eq:VIEE} for
$p=\pm 1$, and we can re-write \eqref{eq:Maine} in more familiar form as
\be
\lim_{t \to 0}t^{2H}\log \Ex(e^{\frac{p}{t^{\al}} X_t})
= \lim_{t \to 0} t^{2H}\log \Ex(e^{\frac{p}{t^{2H}} \frac{X_t}{t^{\half-H}}}) =\left\{
            \begin{array}{ll}
      \bar{\Lm}(p)   \quad \quad  p\in (p_-,p_+)\\
+\infty \quad \quad\,\,  p\notin (p_-,p_+) \nn
\end{array}\
\right.
\ee
where $p_{\pm}=\pm (T^*(\pm 1))^{\al}$, so $p_+>0$ and $p_-<0$.  Then $X_t/t^{\half-H}$ satisfies the LDP as $t \to 0$ with speed $t^{-2H}$ and good rate function $I(x)$ equal to the Fenchel-Legendre transform of $\bar{\Lm}$.
\end{thm}

\nind \begin{proof}
 We first consider the following family of re-scaled Rough Heston models:
\begin{align}
dX^{\e}_t = -\half \e V^{\e}_t dt+  \sqrt{\e} \sqrt{V^{\e}_t} dB_t \quad ,\quad V^{\e}_t =
V_0\,+\,\frac{\e^{\al}}{\Gamma(\al)} \int_0^t (t-s)^{H-\half}  \lm(\theta-V^{\e}_s) ds  +\frac{\e^H}{\Gamma(\al)}\int_0^t (t-s)^{H-\half}  \nu \sqrt{V^{\e}_s}  dW_s \label{eq:Xepst}
\end{align}
with $X^{\e}_t=0$, where $H=\al-\half \in (0,\half]$.  Then from Appendix \ref{section:AppB} we know that
\begin{equation}\label{eq:scaling X}
  (X^{\e}_{(.)},V^{\e}_{(.)})\eqd (X_{\e (.)},V_{\e(.)})\,
\end{equation}
(note this actually holds for all $\lm>0$, but we are only considering $\lm=0$ in this proof).  Proceeding along similar lines to Theorem 4.1 in \cite{FZ17}, we let $\tilde{X}^{\e}_t$ denote the solution to
\be
d\tilde{X}^{\e}_t = \sqrt{\e} \sqrt{V^{\e}_t} dB_t\label{eq:Xtilde}
\ee
with $\tilde{X}^{\e}_0=0$.  From Eq 8 in \cite{ER18} we know that
\be \Ex(e^{p \tilde{X}_t}) = \Ex^{\Qb_p}(e^{\half p^2 \int_0^t V_s ds})\nn\ee
where $\tilde{X}_t:=\tilde{X}^1_t$ and $\Qb_p$ is defined as in \cite{ER18}, but under $\Qb_p$ the value of the mean reversion speed changes from zero to $\bar{\lm}=\rho p \nu$, so
\be
\Ex(e^{p \tilde{X}_t})=e^{V_0 I^{1-\al}\psi(p,t)}\nn
\ee
on some non-empty interval $[0,T^*(p))$, where
\be
D^{\al}\psi(p,t) = \half p^2 +   p \rho \nu \psi(p,t) +\half \nu^2  \psi(p,t)^2 \nn \,
\ee
with $\psi(p,0)=0$.  Existence and uniqueness of solutions to these kind of fractional differential equations (FDE) is standard, as is their equivalence to VIEs, see e.g.\ \cite{GGP19} and chapter 12 of \cite{GLS90} for details.    \\
From Propositions 2 and 3 in \cite{GGP19}, we know that $\psi(p,t)$ blows up at some finite time $T^*(p)>0$ (i.e. case A or B in the \cite{GGP19} classification).Thus we see that
\begin{align}
  \Ex(e^{\frac{p}{\e^\al}\tilde{X}_t^\e})=\Ex(e^{\frac{p}{\e^\al}\tilde{X}_{\e t}})
  =e^{ V_0 I^{1-\al} \psi(\frac{p}{\e^{\al}},\e t)}
  =e^{\frac{1}{\e^{2H}} V_0 I^{1-\al} \psi(p,t)}\label{eq:exxact}
\end{align}
for all $t\in [0,T^*(p))$, which we can re-write as $
  \Ex(e^{\frac{p}{t^\al}\tilde{X}_t})
  =e^{\frac{\bar{\Lm}(p)}{t^{2H}} }.\nn
$
Thus we see that
\begin{align*}
  \lim_{t\to 0}t^{2H} \log \Ex(e^{\frac{p}{t^\al}\tilde{X}_t})=
  \bar{\Lm}(p)\nn
\end{align*}
and $\Lm(p):=\Lm(p,1)<\infty$
 if and only if $T^*(p)>1$.  \\ 

We now have the following obvious but important corollary of the $\Lambda$ scaling relation in \eqref{eq:SR}:
\begin{cor}\label{cor:SR}
\be
\Lm(q) = t^{2H}\Lm(\frac{q}{t^{\al}},t)
= |q|^{\frac{2H}{\al}}\Lm(\mathrm{sgn}(q),|q|^{\frac{1}{\al}})
\label{eq:SR2}
\ee
where we have set $p=1=\frac{|q|}{t^{\al}}$ in \eqref{eq:SR}, and $t_q^*=|q|^{\frac{1}{\al}}$.
\end{cor}



\begin{rem}
This implies that $\Lm(p)\to \infty$ as $p \to p_{\pm}:=\pm (T^*(\pm 1))^{\al}$.
and more generally \begin{equation}\label{eq:T}
p T^*(p)^{\al} = 1_{p>0}\,p_{+}\,+\,1_{p<0}\,p_{-}\,.
\end{equation}
\end{rem}
 To prove the LDP, we first prove the corresponding LDP for $\tilde{X}_t$.  From Lemma 2.3.9 in \cite{DZ98}, we know that
 \[
 \lim_{t\to0}t^{2H} \log \Ex(e^{\frac{p}{t^{\al}} \tilde{X}_t})=\Lm(p)=\Lm(p,1)=I^{1-\al}\psi(p,t)|_{t=1}
 \]
 is convex in $p$, and from \eqref{eq:SP} and \eqref{eq:RiccEq} we know that
\be
\frac{d}{dt}\Lm(p,t) =\half  p^2+ p \rho \nu \psi(p,t) +\half \nu^2  \psi(p,t)^2 \nn
\ee
(where we have also used that $D^{\al} D^{1-\al}=D$), which shows that $\Lm(p,t)$ is also differentiable in $t$, and thus from \eqref{eq:SR2}, we see that $\Lm(p)=\Lm(p,1)$ is differentiable in $p$ for $p>0$.  Moreover the scaling relation easily yields that $\Lm(p)$ is right differentiable at $p=0$, since $\Lm(p)=o(p)$.  We also know that $\psi(p,t)\to \infty$ as $t\to T^*(p)$ (see Propositions 2 and 3 in \cite{GGP19}), so $\Lambda(p,t)=I^{1-\al}\psi(p,t)$ also explodes at $T^*(p)$ by Lemma 3 in \cite{GGP19}.
Then from Corollary \ref{cor:SR}, we know that $\Lm(p)=p^{\frac{2H}{\al}}\Lm(\mathrm{sgn}(p),|p|^{\frac{1}{\al}})$,
so $\Lm(p)\to \infty$ as $p\to p_{\pm}=\pm (T^*(\pm 1))^{\al}$ and (by convexity and differentiability) $\Lm$ is also essentially smooth, so by the G\"{a}rtner-Ellis theorem from large deviations theory (see Theorem 2.3.6 in \cite{DZ98}), $\tilde{X}^{\e}_1/\e^{\half-H}$ satisfies the LDP as $\e \to 0$ with speed $\e^{-2H}$ and rate function $I(x)$.

\sk
\sk

We now show that $X^{\e}_1/\e^{\half-H}$ satisfies the same LDP, by showing that the non-zero drift of the log stock price can effectively be ignored at leading order in the limit as $\e\to 0$.  Using that 
\be
 \Ex(e^{ \frac{p}{\e^{2\al}} \e \int_0^1 V^{\e}_s ds})
= \Ex(e^{\frac{p}{\e^{2H}} \int_0^1 V^{\e}_s ds})
~ \Ex(e^{\frac{\sqrt{2 p}}{\e^{\al}} \tilde{X}^{\e}_1})=
e^{\frac{1}{\e^{2H}}V_0 \Lm(\sqrt{2p})}\nn
\ee
for $p \in (-\infty,\half p_+)$ (and $+\infty$ otherwise) so
\be
J(p) :=\lim_{\e\to0} \e^{2H} \log \Ex(e^{ \frac{p}{\e^{2\al}} \e \int_0^1 V^{\e}_s ds})
 ~V_0\Lm(\sqrt{2p})\nn
\ee
so (again using part a) of the G\"{a}rtner-Ellis theorem in Theorem 2.3.6 in \cite{DZ98}), $A_{\e}:=\int_0^1 V^{\e}_s ds$ satisfies the upper bound LDP as $\e \to 0$ with speed $\e^{-2H}$ and good rate function $J^*$ equal to the FL transform of $J$.
\sk
But we also know that
\be
X^{\e}_1-\tilde{X}^{\e}_1 = -\half \e A_{\e} \nn
\ee
and for any $a>0$ and $\delta_1>0$
\be
\Pb(|\frac{X^{\e}_1}{\e^{\half-H}}-\frac{\tilde{X}^{\e}_1}{\e^{\half-H}}| >\delta)
=\Pb(\half \e^{\half+H} A_{\e}>\delta) = \Pb( A_{\e}>\frac{2\delta}{\e^{\half+H}})
\le \Pb( A_{\e}>a) \lee e^{-\frac{\inf_{a'\ge a} J^*(a')-\delta_1}{\e^{2H}}}
\nn
\ee
for any $\e$ sufficiently small, where we have use the upper bound LDP for $A_{\e}$
to obtain the final inequality.  Thus
\be
\limsup_{\e \to 0}\e^{2H}\log\Pb(|\frac{X^{\e}_1}{\e^{\half-H}}-\frac{\tilde{X}^{\e}_1}{\e^{\half-H}}| >\delta)
\le - \inf_{a'>a}J^*(a') \nn \,
\ee
but $a$ is arbitrary and (from Lemma 2.3.9 in \cite{DZ98}), $J^*$ is a good rate function, so in fact
\be
\limsup_{\e \to 0}\e^{2H}\log\Pb(|\frac{X^{\e}_1}{\e^{\half-H}}-\frac{\tilde{X}^{\e}_1}{\e^{\half-H}}| >\delta)
=-\infty \nn \,.
\ee
Thus $\frac{X^{\e}_1}{\e^{\half-H}}$ and $\frac{\tilde{X}^{\e}_1}{\e^{\half-H}}$ are \textit{exponentially equivalent} in the sense of Definition 4.2.10 in \cite{DZ98}, so (by Theorem 4.2.13 in \cite{DZ98})
$\frac{X^{\e}_1}{\e^{\half-H}}$ satisfies the same LDP as $\frac{\tilde{X}^{\e}_1}{\e^{\half-H}}$.
\nind \end{proof}


\subsection{Asymptotics for call options and implied volatility}

\sk
\sk

\begin{cor}
We have the following limiting behaviour for out-of-the-money European put and call options with maturity $t$ and log-strike $t^{\half-H}x$, with $x\in \mathbb{R}$ fixed:
\begin{align}\label{eq:call}
\lim_{t\to 0} t^{2H}\log\Ex((e^{X_t}-e^{x t^{\half-H}})^+)
= -I(x) \quad \quad (x>0)\nn \\
\lim_{t\to 0} t^{2H}\log\Ex((e^{x t^{\half-H}}-e^{X_t})^+)
= -I(x)\quad \quad (x<0)\nn     \,.
\end{align}
\end{cor}
\begin{proof}
 The lower estimate follows from the exact same argument used in Appendix C in \cite{FZ17}
  (see also Theorem 6.3 in \cite{FGP18b}).
  The proof of the upper estimate is the same as in Theorem 6.3 in \cite{FGP18b}.
\end{proof}


\begin{cor}
Let $\hat{\sigma}_t(x)$ denote the implied volatility of a European put/call option with log-moneyness $x$ under the Rough Heston model in \eqref{eq:Vt} for $\lm=0$.  Then for $x \ne 0$ fixed, the implied volatility satisfies
  \begin{equation}\label{eq:imp vol}
    \hat{\sigma}(x):=\lim_{t\to0}\hat{\sigma}_t(t^{\half-H}x) = \frac{|x|}{\sqrt{2I(x)}}.
  \end{equation}
\end{cor}
\begin{proof}
Follows from Corollary 7.2 in \cite{GL14}. See also the proof of Corollary 4.1 in \cite{FGP18b}
for details on this, but the present situation is simpler, as we only require the leading order term here.
\end{proof}


\subsection{Series expansion for the asymptotic smile and calibration}
\label{ss:ss}

\sk
\sk
Proceeding as in Lemma 12 in \cite{GGP19}, we can compute a fractional power series for $\psi(p,t)$
 (and hence $\Lm(p,t)$) and then using \eqref{eq:SR2}, we find that
\be
\bar{\Lm}(p) = \frac{2V_0}{\nu^2}\sum_{n=1}^{\infty} a_n(1) p^{1+n}  \frac{\Gm(\al n + 1)}{\Gamma(2+(n-1)\al )}\nn\label{eq:LmSeries}
\ee
where the $a_n=a_n(u)$ coefficients are defined (recursively) as in \cite{GGP19}
except for our application here (based on \eqref{eq:RiccEq}) we have to set $\lm=0$, and $c_1=\half u^2$ instead of $\half u(u-1)$ (note this series will have a finite radius of convergence).  Using the Lagrange inversion theorem, we can then derive a power
series for $I(x)$ which takes the form
\be
\hat{\sigma}(x) =\sqrt{V_0} + \frac{\rho \nu }{2 \Gamma(2 + \al) \sqrt{V_0}}x +
\nu^2 \frac{\Gamma(1 + 2 \al) +
   2 \rho^2 \Gamma(
     1 + \al)^2 (2 -
      3 \frac{\Gamma(2 + 2 \al)}{\Gamma(2 + \al)^2} ) }{8 V_0^{\frac{3}{2}}
  \Gamma(1 + \al)^2 \Gamma(2 + 2 \al)}x^2+O(x^3) \,.\label{eq:ISeries}
\ee
(compare this to Theorem 3.6 in \cite{BFGHS18} for a general class of rough models and Theorem 4.1 in \cite{FJ11b}
for a Markovian local-stochastic volatility model).  We can re-write this expansion more concisely in dimensionless form as
\be
\hat{\sigma}(x) =\sqrt{V_0}\,[1 + \frac{ \rho }{2 \Gamma(2 + \al) }z +
 \frac{\Gamma(1 + 2 \al) +
   2 \rho^2 \Gamma(
     1 + \al)^2 (2 -
      3 \frac{\Gamma(2 + 2 \al)}{\Gamma(2 + \al)^2} ) }{8
  \Gamma(1 + \al)^2 \Gamma(2 + 2 \al)}z^2+O(z^3)] \nn\,\label{eq:ISeries2}
\ee
where the dimensionless quantity $z=\frac{\nu x}{V_0}$.

\begin{rem}
In principle one can use \eqref{eq:ISeries} to calibrate $V_0$, $\rho$ and $\nu$ to observed/estimated values
of $\hat{\sigma}(0)$, $\hat{\sigma}'(0)$ and $\hat{\sigma}''(0)$ (i.e. the short-end implied vol level, skew and convexity respectively).
\end{rem}

\subsubsection{Wing behaviour of the rate function}

\sk

From Eq 3.2 in \cite{RO96}, we expect that
$
  \psi(p,t) \sim \frac{const.}{(T^*(p)-t)^\al}
$ as $t \to T^*(p)$
and thus
$
  \Lambda(p,t) =  I^{1-\al} \psi(p,t) \sim
  \frac{const.}{(T^*(p)-t)^{2\al-1}}\,\nn
$
as $t \to T^*(p)$.  Assuming this is consistent with the $p$-asymptotics, then (by \eqref{eq:T}) we have
\[
  \Lambda(p)=\Lambda(p,1) \sim \frac{const.}{(T^*(p)-1)^{2\al-1}}
  = \frac{const.}{((\frac{p_+}{p})^{1/\al}-1)^{2\al-1}}
  \sim \frac{const.}{(p_+-p)^{2\al-1}} \quad \quad (p \to p_+)
\]
so $p^*(x)$ in
$I(x)=\sup_p (px - V_0 \Lambda(p))
$
satisfies $p^*(x)= p_+ -const. \cdot x^{-1/2\al}(1+o(1))$, so
$
I(x) = p_+x + const.\cdot x^{1-\frac{1}{2\al}}(1+o(1))$ as $x \to \infty$.

\sk

\subsection{Higher order Laplace asymptotics}
\label{subsection:HigherOrder}

\sk
If we now relax the assumption that $\lm =0$, and work with the original $X^{\e}$ process in \eqref{eq:Xepst} (as opposed to the driftless $\tilde{X}^{\e}$ process in \eqref{eq:Xtilde}), then we know that
\be
\Ex(e^{p X^{\e}_t})=\Ex(e^{p X_{\e t}})=e^{V_0 I^{1-\al}g_{\e}(p,t)+\e^{\al}\lm \theta I^1 g_{\e}(p,t)}\nn
\ee for $t$ in some non-empty interval $[0,T^*_{\e}(p))$, where
\be
g_{\e}(\frac{p}{\e^{\al}},t) = \frac{\psi(p,t)}{\e^{2H}}\label{eq:orbit}\,\ee which satisfies
\begin{align}
D^{\al}g_{\e}(p,t)
&=\half \e (p^2-p) +   (p \rho \nu-\lm)  \e^{\al} g_{\e}(p,t) +\half \e^{2H}\nu^2  g_{\e}(p,t)^2 \label{eq:FDE2}
\end{align}
with initial condition $g_{\e}(p,0)=0$.  Setting
\be
g_{\e}(\frac{p}{\e^{\al}},t) = \frac{\psi_{\e}(p,t)}{\e^{2H}}\label{eq:orbit2}\,\ee
and setting $p\mapsto\frac{p}{\e^{\al}}$, and substituting for $g_{\e}(\frac{p}{\e^{\al}},t)$ in \eqref{eq:FDE2} and multiplying by $\e^{2H}$ as before, we find that
\be
D^{\al}\psi_{\e}(p,t) =
\half  p^2+   p \rho \nu   \psi_{\e} (p,t) +\half \nu^2  \psi_{\e}(p,t)^2  -\e^{\al}(\half p +\lm \psi_{\e} (p,t)) \nn
\nn \ee
\sk
\nind with $\psi_{\e}(p,0)=0$.  If we now formally try a higher order series approximation of the form $\psi_{\e}(p,t):=\psi(p,t)+\e^{\half+H} \psi_1(p,t)$, we find that
$\psi_1(p,t)$ must satisfy
\be
D^{\al}\psi_1(p,t) = -\half p\,-\lm \psi(p,t) + p  \rho \nu \psi_1(p,t) +\nu^2\psi(p,t)\psi_1(p,t) \,\nn
\ee
with $\psi_1(p,0)=0$, which is a linear VIE for $\psi_1(p,t)$.

\begin{rem}
Let $\Delta_{\e}(p,t)=\psi_{\e}(p,t)-\psi(p,t)-\e^{\half+H} \psi_1(p,t)$ denote the error term.  Then
$\Delta_{\e}(p,t)$ satisfies
\begin{align}
D^{\al}\Delta_{\e}(p,t) \,&=\,
p \nu \rho \Delta_{\e}(p,t) \,+\,
\half \nu^2 \Delta_{\e}(p,t)^2 \,+\, \nu^2 \Delta_{\e}(p,t) \psi(p,t)\,+\,\nn \\
&+ \, \e^{\half +H} \Delta_{\e}(p,t) (-\lm \,+\, \nu^2  \psi_1(p, t)) \nn \\
&+\, \e^{2 H +
     1} (-\lm \psi_1(p, t) \,+\,
    \half \nu^2 \psi_1(p, t)^2)\nn
\nn \end{align}
and the re-scaled error $\bar{\Delta}_{\e}(p,t):=\Delta_{\e}(p,t)/\e^{\half+H}$ satisfies
\begin{align}
D^{\al}\bar{\Delta}_{\e}(p,t) \,&=\,
\bar{\Delta}_{\e}(p,t)  (p \nu \rho  \,+\, \nu^2 \psi(p,t))\,+\,\nn \\
&+ \, \e^{\half +H}  (-\lm \psi_1(p, t) \,+\, \half \nu^2 \psi_1(p, t)^2+(-\lm \,+\, \nu^2  \psi_1(p, t))\bar{\Delta}_{\e}(p,t) + \half \nu^2 \bar{\Delta}_{\e}(p,t)^2) \nn
\nn \end{align}
We know that $\psi(p,t)$ is continuous on $[0,T^*(p))$.  In order to make this rigorous, one would need to apply \cite{GLS90} to this, noting that the leading order solution is zero, then replace $p$ with $i k$ for $k$ real, then show this convergence is uniform on compact sets, and then argue away the tails as in \cite{FJL12}.

\end{rem}

\begin{rem}\label{rem:Gop}
Setting $\psi_1(p,t)=\sum_{n=1}^{\infty} b_n(p) t^{\al n}$ we see that
\bq
\sum_{n=1}^{\infty}  \frac{n \al \Gamma(n \al)}{\Gamma(
 1 + (n-1) \al)}b_n(p) t^{(n-1)\al} &=&-\half p\,-\lm \sum_{n=1}^{\infty} \bar{a}_n(p) t^{\al n}+
 p\rho \nu \sum_{n=1}^{\infty} b_n(p) t^{\al n} +\nu^2\sum_{n=1}^{\infty} \bar{a}_n(p) t^{\al n}
 \sum_{m=1}^{\infty} b_m(p) t^{\al m} \nn \,
\eq
where $\bar{a}_n(p)=\frac{2}{\nu^2}a_n(p)$, and we have set $\lm=0$ and $c_1=\half p^2$ in computing the $a_n(p)$ coefficients, so
\be
\al \Gamma( \al)
 b_1(p)  = -\half p\quad ,\quad
\frac{(n+1) \al \Gamma((n+1) \al)}{\Gamma(
 1 + n \al)} b_{n+1}(p) = -\lm \bar{a}_n(p) + \rho p \nu b_n(p)+\nu^2 \sum_{k=1}^{n-1} a_k(p) b_{n-k}(p) \nn\,
\ee
so we have fractional power series for $\psi_1(p,t)$ on some finite radius of convergence.
\end{rem}

Returning now to the main calculation, we see that if $p_{\e}(x)$ denotes the density of $\frac{X_1^{\e}}{\e^{\al}}$, then
\be
p_{\e}(\frac{x}{\e^{2H}}) = \frac{1}{2\pi}\int_{-\infty}^{\infty} e^{-\frac{ikx}{\e^{2H}}} e^{\frac{1}{\e^{2H}}(F(k)+\e^{\half+H} G(k))+\frac{\e^{\al}}{\e^{2H}}\lm \theta (F_1(k)+\e^{\half+H}  G_1(k))} dk\nn \,
\ee
where $F(k):=V_0 I^{1-\al}\psi(ik,1)$, $G(k):=V_0 I^{1-\al}\psi_1(ik,1)$, $F_1:=I^1 \psi(ik,1)$ and $G_1:=I^1 \psi_1(ik,1)$.  The saddlepoint $k^*=k^*(x)=ip^*(x)$ of
$\bar{F}(k)=-ikx+ F(k)$ satisfies $\bar{F}'(k^*)=0$ which always falls on the imaginary axis (and in our case $p^*(x) \in (0,p_+)$ when $x>0$ and $p^*(x)<0 \in (p_-,0)$ when $x<0$), and
\begin{align*}
\bar{F}(k) &= \bar{F}(k^*)+\half F''(k^*)(k-k^*)^2 + O((k-k^*)^3) \\
&= \bar{F}(k^*)-\half \bar{\Lm}''(p^*)(k-k^*)^2+O((k-k^*)^3) \nn
\end{align*}
(recall that $\bar{\Lm}(p)=F(-ip)$) and $p^*=ik^* \in (p_-,p_+)$.  Then proceeding along similar lines to \cite{FJL12} and using Laplace's method we have for all $x\in \mathbb{R}$
\begin{align}
p_{\e}(\frac{x}{\e^{2H}}) &= \frac{1}{2\pi}\int_{-\infty}^{\infty} e^{\frac{1}{\e^{2H}} (\bar{F}(k)+\e^{\half+H}  G(k))+\e^{\half-H}\lm \theta (F_1(k)+\e^{\half+H}  G_1(k))} dk\label{eq:lap dens1} \\
&\approx \frac{1}{2\pi} e^{\e^{\half-H}(G(k^*)+\lm \theta  F_{1}(k^*))}\int_{-\infty}^{\infty}  e^{\frac{1}{\e^{2H}}(\bar{F}(k^*)-\half \bar{\Lm}''(p^*)(k-k^*)^2)} dk\label{eq:lap dens2} \\
&\approx \frac{1}{2\pi}e^{\e^{\half-H}(G(k^*)+\lm \theta  F_{1}(k^*))} e^{-\frac{I(x)}{\e^{2H}}}\int_{-\infty}^{\infty}  e^{-\frac{1}{\e^{2H}}\half \bar{\Lm}''(p^*)(k-k^*)^2} dk \nn \\
&=  \frac{\e^{H}e^{-\frac{I(x)}{\e^{2H}}}}{ \sqrt{2\pi \bar{\Lm}''(p^*)}}[1+\e^{\half-H}(G(k^*)+\lm \theta  F_{1}(k^*))+O(\e^{(1-2H) \wedge 2H})]\label{eq:lap dens3}
\end{align}
where the $O(\e^{2H})$ part of the error terms comes from the next order term in Theorem 7.1 in chapter 4 in \cite{Olv74}, and the $\e^{(1-2H)}$ term comes from the 2nd order term in expanding the exponential.
The meaning of $\approx$ in the above estimates is as follows: we expect to
have asymptotic equality with a relative error term that does not interfere
with the error term in \eqref{eq:lap dens3}, but since we did not carry
out the tail estimate of the saddle point approximation, we do not know its size.
 Then letting $z=\frac{k}{\e^{\al}}$, we see that
\begin{align}
C_{\e}(x)=\Ex((e^{X_1^{\e}}-e^{x \e^{\half-H}})^+)&=\frac{1}{2\pi}e^{x\e^{\half-H}}\int_{-ip^*-\infty}^{-ip^*+\infty} \mathrm{Re}(\frac{e^{-izx\e^{\half-H}}}{-iz-z^2} \Ex(e^{iz X_1^{\e}})) dz\nn \\
&=\frac{1}{2\pi}e^{x\e^{\half-H}}\int_{-ip^*-\infty}^{-ip^*+\infty} \mathrm{Re}(\frac{e^{-i\frac{k}{\e^{2H}}x }}{-i\frac{k}{\e^{\al}}-(\frac{k}{\e^{\al}})^2}\, \Ex(e^{i\frac{k}{\e^{\al}} X_1^{\e}}))  d\frac{k}{\e^{\al}}\label{eq:numm} \\
&=\frac{\e^{-\al}}{2\pi}e^{x\e^{\half-H}}\int_{-ip^*-\infty}^{-ip^*+\infty} \mathrm{Re}(e^{i\frac{k}{\e^{2H}}x}
 (- \frac{\e^{2\al}}{k^2}-i\frac{\e^{3\al}}{k^3} +O(\e^{4\al}))
  \,   \Ex(e^{i\frac{k}{\e^{\al}} X_1^{\e}}))dk \label{eq:Lap call} \\
&=  \frac{\e^{\half+2H}e^{-\frac{I(x)}{\e^{2H}}}}{(p^*)^2 \sqrt{2\pi \bar{\Lm}''(p^*)}}\,[1+\e^{\half-H}(x+G(k^*)+\lm \theta  F_{1}(k^*))+O(\e^{(1-2H) \wedge 2H})]\nn \\
&=  \frac{A(x)\e^{\half+2H}e^{-\frac{I(x)}{\e^{2H}}}}{\sqrt{2\pi}}[1+\e^{\half-H}(x+G(k^*)+\lm \theta  F_{1}(k^*))+O(\e^{(1-2H) \wedge 2H})]
\label{eq:LaplaceCalls}
\end{align}
where
\bq
A(x)&=& \frac{1}{(p^*)^2 \sqrt{ \bar{\Lm}''(p^*)}}\label{eq:A((x))}
\eq

\sk

The $\e$-dependence of the leading order term here is exactly the same as in Corollary 7.1 in the recent article of Friz et al.\ \cite{FGP18a} (in \cite{FGP18a} $\e^2=t$ whereas here $\e=t$) which deals with a general class of rough stochastic volatility models (which excludes Rough Heston).
 The difficulty in making the expansions \eqref{eq:lap dens3}
 and \eqref{eq:LaplaceCalls} rigorous is the step from \eqref{eq:Lap call}
 to \eqref{eq:LaplaceCalls}, or, more explicitly, from \eqref{eq:lap dens1} to
 \eqref{eq:lap dens2}. The expansion of $\bar F$ used in \eqref{eq:lap dens2} is valid
 locally, close to the saddle point. An estimate for the integrand in \eqref{eq:lap dens1}
 is needed to argue that this is good enough, i.e., that the asymptotic
 behavior of \eqref{eq:lap dens1} is captured by integrating over an appropriate
 neighbourhood of the saddle point. This is usually done by establishing
 monotonicity of the integrand, but seems non-trivial here; cf.\ Lemma 6.4 in  \cite{FJL12},
 which uses the explicit characteristic function of the classical Heston model.

 \sk
 \sk
 More generally, we can formally substitute a fractional power series of the form $\psi_{\e}(p,t)=\sum_{n=0}^{\infty}\psi_n(p,t)\e^{(n+1)\al}$ (where $\psi_0(p,t):=\psi(p,t)$), and we find that $(\psi_n)_{n\ge 1}$ satisfies a nested sequence of linear fractional differential equations:
\begin{align}
D^{\al}\psi_1(p,t) &= -\half p\,-\lm \psi_0(p,t) + p  \rho \nu \psi_1(p,t) +\nu^2\psi_0(p,t)\psi_1(p,t) \,\nn \\
D^{2\al}\psi_2(p,t) &= -\lm \psi_1(p,t) + p  \rho \nu \psi_2(p,t) +\nu^2\psi_0(p,t)\psi_2(p,t)+\half \nu^2 \psi_1(p,t)^2 \,\nn \\
&... \nn \\
 D^{n\al}\psi_n(p,t) &= \,-\lm \psi_{n-1}(p,t) + p  \rho \nu \psi_n(p,t) +\half \nu^2 [\sum_{k=0}^n \psi_k(p,t)\psi_{n-k}(p,t)
 +  1_{\half n \in \mathbb{N}}\cdot\psi_{\half n}(p,t)^2] \,
\end{align}
with $\psi_n(p,0)=0$, and in principle we can then compute fractional power series expansions for each $\psi_n(p,t)$ of the form
$\psi_n(p,t)=\sum_{m=1}^{\infty} a_{m,n}(p) t^{\al m}$, as in Remark \ref{rem:Gop} above.

\subsubsection{Higher order expansion for implied volatility}

\sk

\textbf{Formal corollary of \eqref{eq:LaplaceCalls}}:
Let $\hat{\sigma}_t(x)$ denote the implied volatility of a European put/call option with log-moneyness $x$ under the Rough Heston model in \eqref{eq:Vt} for $\lm\ge 0$.  Then for $x \ne 0$ fixed, the implied volatility satisfies
\begin{equation}\label{eq:imp vol}
\hat{\sigma}_t(t^{\half-H}x)^2 = \frac{|x|}{\sqrt{2I(x)}}+t^{2H} \Sigma_1(x)+o(t^{2H})
\end{equation}
where
\be
\Sigma_1(x) = \frac{x^2 \log A_1(x)}{2 I(x)^2}\nn\,
\ee
and
where $A_1(x)=2A(x)I(x)^{\frac{3}{2}}/x$ \footnote{We thank Peter Friz and Paolo Pigato for clarifying the main steps in this result}.
\nind 

\nind \begin{proof}
Let $L_t=-\log C_t(x)$, where $C_t(x)$ is defined as in \eqref{eq:LaplaceCalls}.  Then using Corollary 7.1 and Eq 7.2 in Gao-Lee\cite{GL14} we see that
\bq
|\frac{1}{t}G_{-}^2(k_t,L_t-\frac{3}{2}\log L_t+\log\frac{k_t}{4\sqrt{\pi}},u)-\hat{\sigma}_t^2(k_t)| &=&
o(\frac{k_t^2}{L_t^2 t})\nn
\eq
where $G_{-}(k, u):=\sqrt{2} (\sqrt{u + k} - \sqrt{u})$.  Then
\begin{align}
L_t-\frac{3}{2}\log L_t+\log\frac{k_t}{4\sqrt{\pi}}
&=\frac{I(x)}{t^{2H}}-(\half+2H)\log t -\log \frac{A(x)}{\sqrt{2\pi}} -\frac{3}{2}\log(\frac{I(x)}{t^{2H} }(1-\log A(x) \frac{t^{2H}}{I(x)})) \nn \\
&+ \log\frac{x}{4\sqrt{\pi}}\,+\, (\half-H)\log t \nn
\end{align}
where $A(x)$ is defined as in \eqref{eq:A((x))}.  Collecting $\log t$ terms we find that their sum vanishes, so
\begin{align}
L_t-\frac{3}{2}\log L_t+\log\frac{k_t}{4\sqrt{\pi}}
&=\frac{I(x)}{t^{2H}} -\log \frac{A(x)}{\sqrt{2\pi}} -\frac{3}{2}\log I(x) +\log\frac{x}{4\sqrt{\pi}}+o(1) \nn \\
&=\frac{I(x)}{t^{2H}} -\log A_1(x)+o(1)\,. \nn
\end{align}
Then using that
\be
G_{-}^2(k,u) = \frac{k^2}{2u} - \frac{k^3}{4u^2}+O(\frac{k^4}{u^3})
\nn
\ee
as $k/u \to 0$, we obtain the result.
\nind \end{proof}

 \subsubsection{Using these approximations in practice}
 \sk

 \eqref{eq:LaplaceCalls} is of little use in practice, since the leading order Laplace approximation ignores the variation of the function $\frac{1}{k^2}$ in the integrand, and even if we partially take account of this effect by going to next order with Laplace's method using the formula in Theorem 7.1 in chapter 4 in \cite{Olv74} (which we have checked and tried), it still frequently gives a worse estimate that the leading order estimate $\hat{\sigma}(x)$ because the higher order error terms being ignored are too large, and since $H$ is usually very small in practice, $t^H$ converges very slowly to zero.  If we instead compute an approximate call price using the Fourier integral along the horizontal contour going through the saddlepoint in \eqref{eq:numm} (using e.g. the \verb"NIntegrate" command in Mathematica) and use our higher order asymptotic estimate $\psi(ik,t)+\e^{\half+H} \psi_1(ik,t)$ for $\log \Ex(e^{i\frac{k}{\e^{\al}} X^{\e}}))$, and then compute the \textit{exact} implied volatility associated with this price (which avoids the problems with the Laplace approximation), then (for the parameters we considered) we found this approximation to be an order of magnitude closer to the Monte Carlo value than the leading order approximation $\hat{\sigma}(x)$ (see graph and tables below). See \cite{LK07} for more on computing the optimal contour of integration for such problems.

 \sk
 \sk


\subsection{Small-time moderate deviations}

\sk
\sk

Inspired by \cite{BFGHS18}, if we replace \eqref{eq:orbit2} with
\be g_{\e}(\frac{p}{\e^{q}},t) = \frac{\psi_{\e}(p,t)}{\e^{2H-2\beta}}\nn\ee
where $q=\half-H+\beta$, then we find that
 \be
D^{\al}\psi_{\e}(p,t) = \half p^2 - \half p \e^{\half-H+\beta} +
 p \e^{-2 H + 3 \beta} \rho \nu \psi_{\e}(p,t) -\e^{\half - 3 H + 4 \beta}\lm  \psi_{\e}(p,t) +
 \half \e^{-4 H + 6 \beta} \nu^2 \psi_{\e}(p,t)^2 \,\nn
\ee
and we see that all non constant terms on the right hand side are $o(1)$ as $\e\to 0$ if $\beta \in (\frac{2}{3}H,H)$ and $H \in (0,\half)$.  Following similar calculations as above, we formally obtain that
$\lim_{t \to 0} t^{2H-2\beta}\log \Ex(e^{\frac{p}{t^{2H-2\beta}} \frac{X_t}{t^{q}}})=V_0 I^{1-\al}I^{\al} (\half p^2) =\half V_0 p^2$ for \textit{all} $p \in \mathbb{R}$, which (modulo some rigour) implies that $X_t/t^q$ satisfies the LDP with speed $\frac{1}{t^{2H-2\beta}}$ and Gaussian rate function $I(x)=\half x^2/V_0$.  Note that $\beta=H$ corresponds to the central limit or Edgeworth regime, see \cite{FSV19} for details.

\begin{figure}
\begin{center}
\includegraphics[width=100pt, height=110pt]{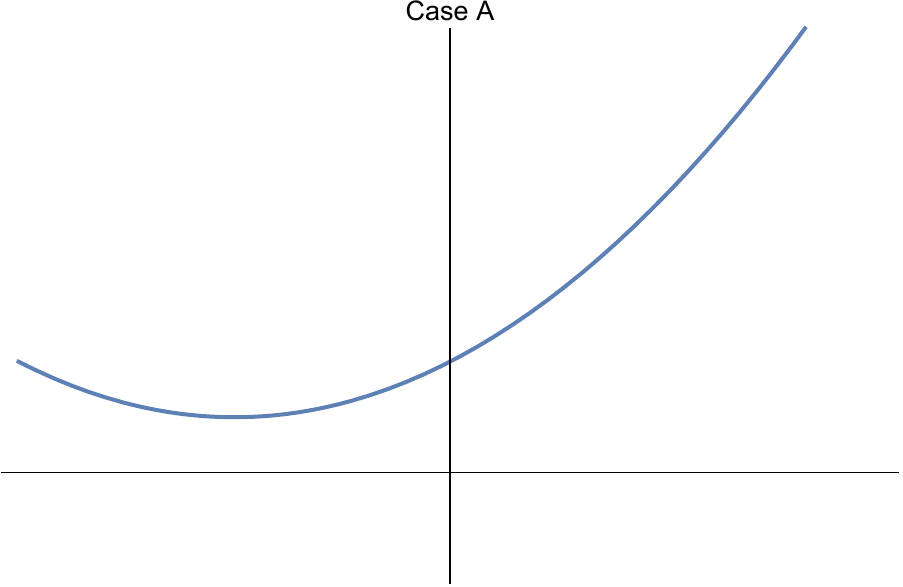}\quad
\includegraphics[width=110pt, height=110pt]{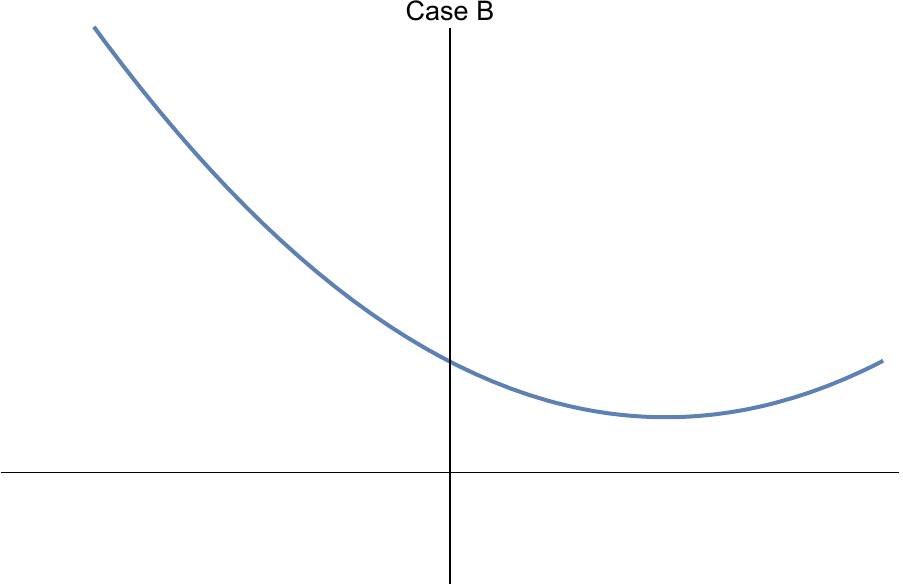}\quad
\includegraphics[width=110pt, height=110pt]{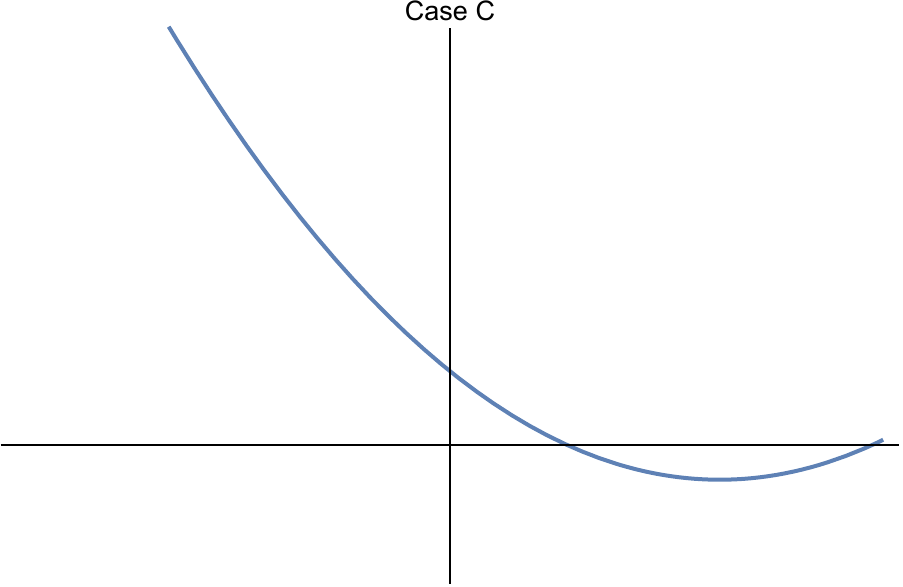}\quad
\includegraphics[width=110pt, height=110pt]{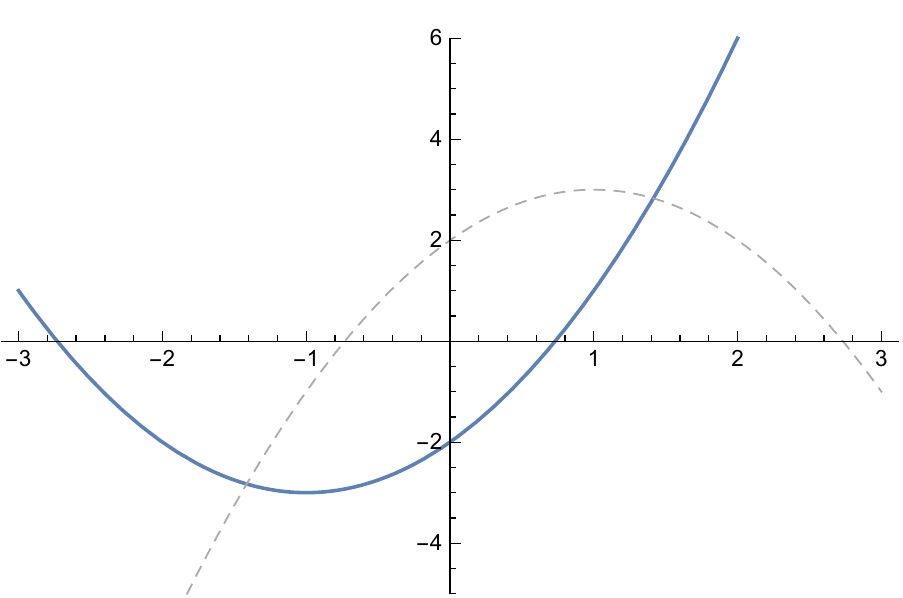}\quad
\caption{Here we have plotted the quadratic function $G(p,w)$ as a function of $w$ for the four cases described in \cite{GGP19}.  In cases A and B there are no roots and the solution $\psi(p,t)$ to \eqref{eq:RiccEq} increases without bound whereas in cases C and D we have a stable fixed point (the lesser of the two roots) and an unstable root, so a solution starting at the origin increases (decreases) until it reaches the stable fixed fixed point. For Case D we have also drawn the curve arising from the reflection transformation used in the proof in Appendix C.  }
\bs
\includegraphics[width=120pt, height=130pt]{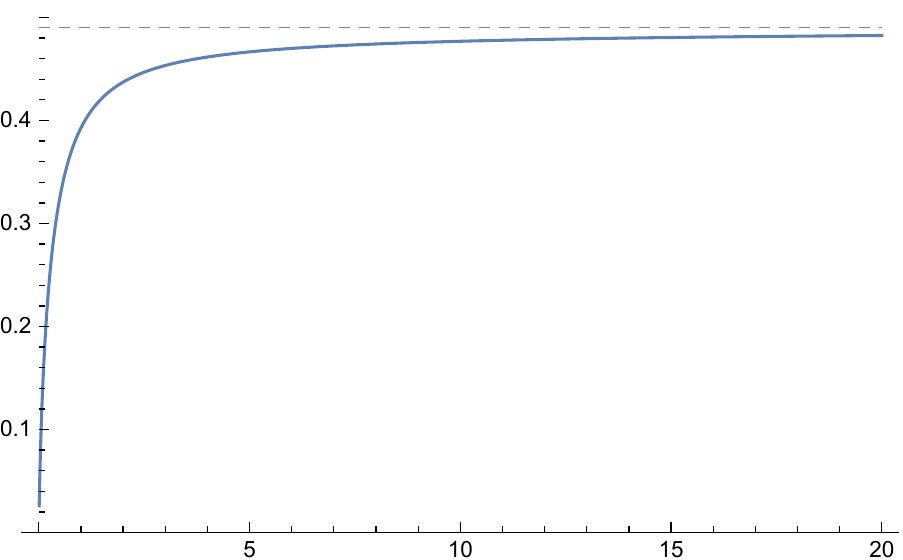} \quad
\includegraphics[width=120pt, height=130pt]{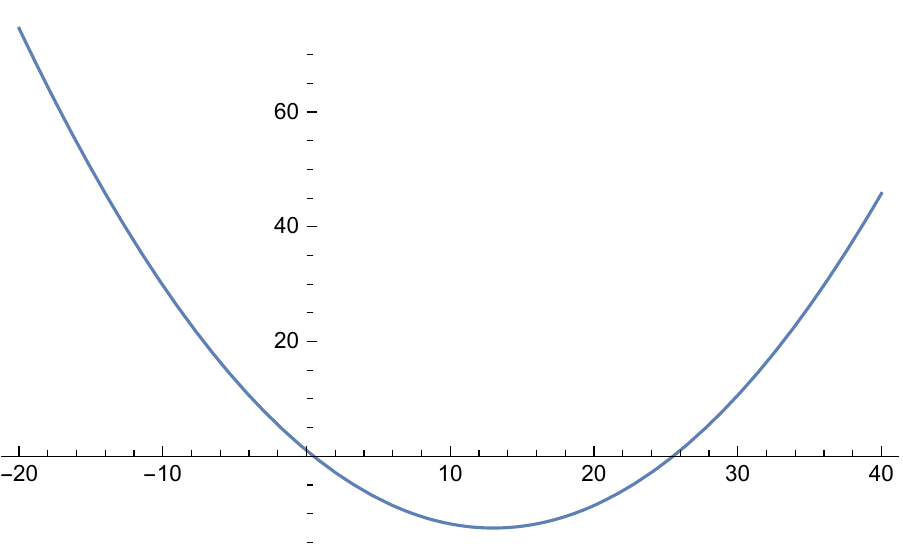} \quad
\includegraphics[width=120pt, height=130pt]{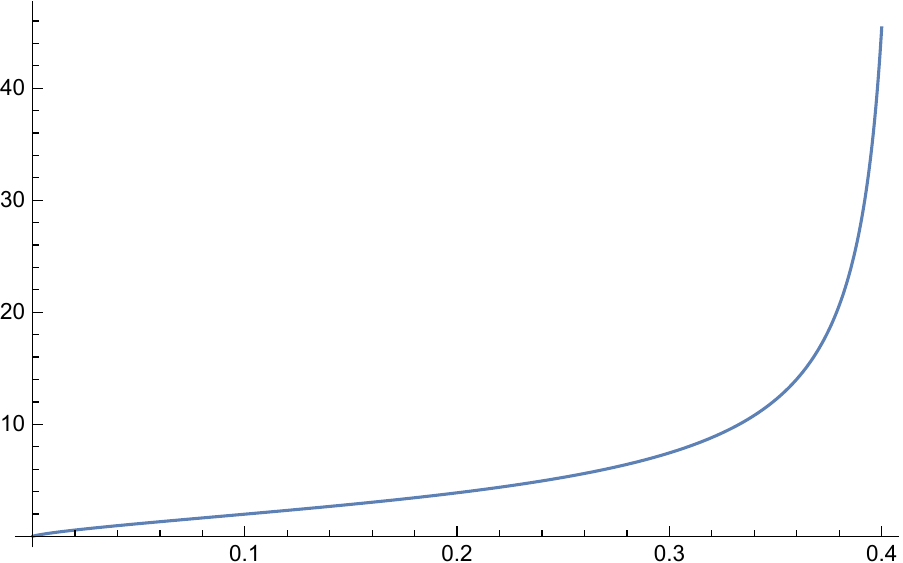}\quad
\includegraphics[width=120pt, height=130pt]{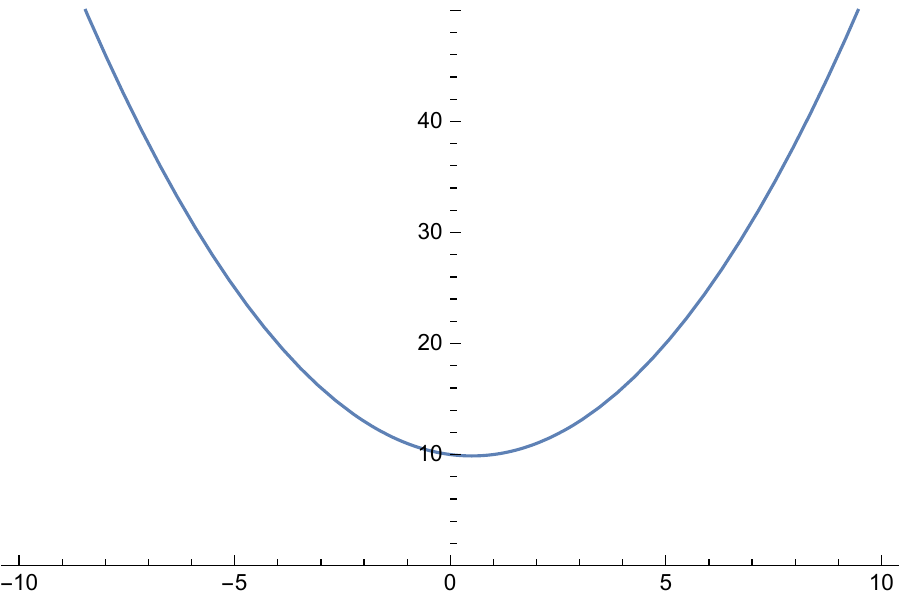}
\caption{Here we have solved for the solution $f(p,t)$ to \eqref{eq:FDE0} numerically by discretizing the VIE with 2000 time steps, and plotted $f(p,t)$ a function of $t$ and the corresponding quadratic function $G(p,w)$ as a function of $w$ with $p$ fixed. In the first case $\al=.75$, $\lm=2$, $\rho=-0.1$, $\nu=.4$ and $p=2$ and $f(p,t)$ tends to a finite constant, and in the second case $\al=.75$, $\lm=1$, $\rho=0.1$, $\nu=1$ and $p=5$ and we see that $f(p,t)$ has an explosion time at some $T^*(p)\approx 0.4$.}
\end{center}
\end{figure}

\begin{figure}
\centering
\begin{center}
\includegraphics[width=180pt, height=180pt]{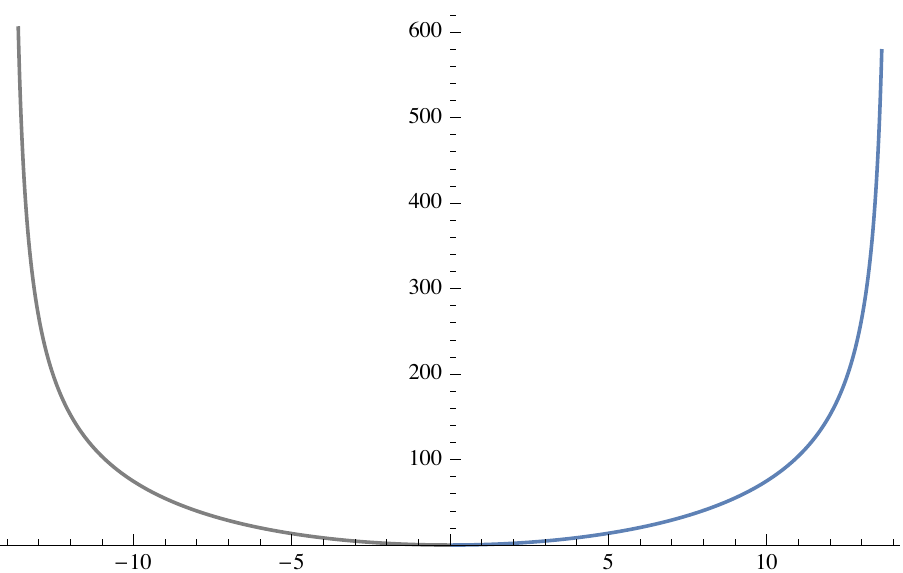}\quad\quad
\includegraphics[width=180pt, height=180pt]{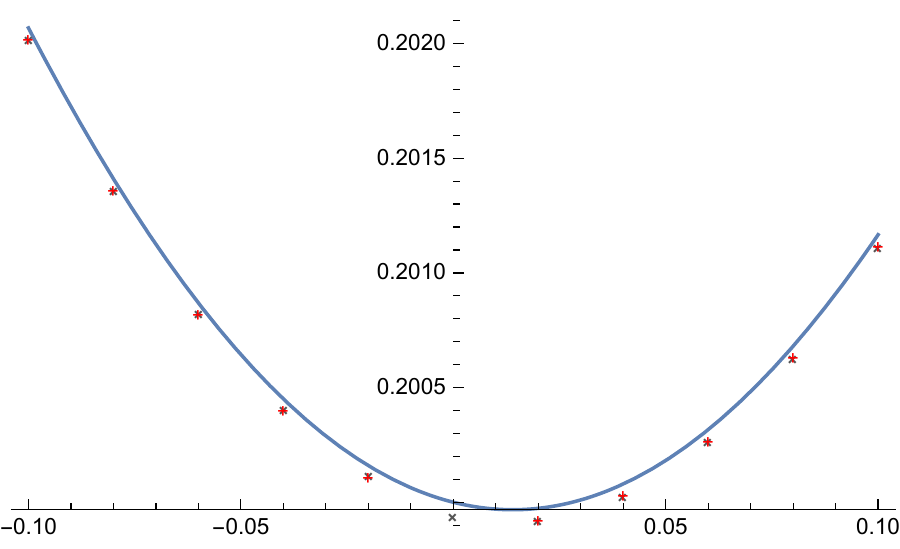}
\caption{On the left we have plotted $\Lm(p)$ using an Adams scheme to numerically solve the VIE in \eqref{eq:RiccEq} with 2000 time steps combined with Corollary \ref{cor:SR}, for $\al=0.75$, $V_0=.04$, $\nu=.15$, $\rho=-0.02$, and we find that $p_+=T^*(1)\approx34.5$ and $p_-=T^*(-1)\approx 33.25$.  On the right we have plotted the corresponding asymptotic small-maturity smile $\hat{\sigma}(x)$ (in blue) verses the higher order approximation using Eq \eqref{eq:numm} (red ``+" signs), and the smile points obtained from a simple Euler-type Monte Carlo scheme with maturity $T=.00005$, $10^5$ simulations and 1000 time steps in Matlab (grey crosses), Matlab and Mathematica code available on request.  We did not use the Adams scheme to compute $\hat{\sigma}(x)$; rather have used the first 15 terms in the series expansion for $\bar{\Lm}(p)$ in subsection \ref{ss:ss} and then numerically computed its Fenchel-Legendre transform and used this to compute $I(x)$ and hence $\hat{\sigma}(x)$.  We see that the Monte Carlo and higher order smile points can barely be distinguished by the naked eye.  For $|x|$ small, we have found this method of computing $\hat{\sigma}(x)$ to be far superior to using an Adams scheme, since the numerical computation of the fractional integral $I^{1-\al}f(p,t)$ for $|t|\ll 1$ can lead to numerical artefacts when computing the FL transform of $\bar{\Lm}(p,1)$ close to $x=0$.
}
\bs
\includegraphics[width=180pt, height=180pt]{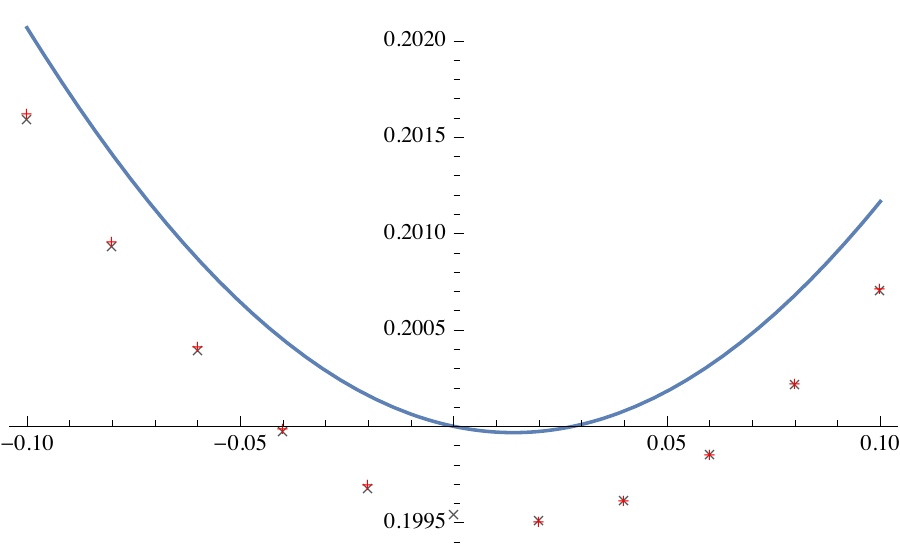} \quad  \quad
\includegraphics[width=180pt, height=180pt]{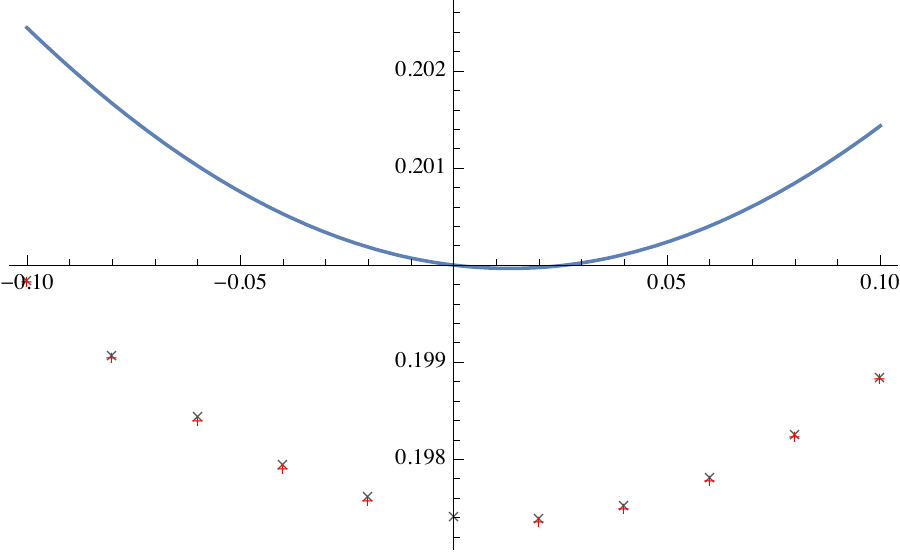}
\caption{On the left here we have the same plot as above but with $T=.005$ and for the right plot $T=.005$
and $\al=.6$ (i.e. $H=0.1$), and again we see that the higher order approximation makes a significant improvement over the leading order smile.  Of course we would not expect such close agreement for smaller values of $\al$, or larger values of $T$, $|x|$ or $|\rho|$, e.g. $\rho= -0.65$ reported in e.g. \cite{EGR18}, but the point here is really just to verify the correctness of the asymptotic formula in \eqref{eq:imp vol}, and give a starting point for other authors/practitioners who wish to test refinements/variants of our formula.  We have not repeated numerical results for the large-time case at the current time, since it is intuitively fairly clear that our large maturity formula is correct (since it just boils down to computing the stable fixed point of the VIE) and for maturities $\approx$ 30 years with a small step-size, the code would take a prohibitively long time to give good results given that each simulation takes $O(N^2)$ for a rough model (where $N$ is the number of time steps), and it is difficult to verify the formula numerically even for the standard Heston model.}
\end{center}
\end{figure}

 \newpage

\begin{center}
\begin{tabular}{ |c|c|c|c|c|c| }
 \hline
$x$ \quad\quad\quad& $\hat{\sigma}(x)$ & \small{Higher order $T=.00005$}   & \small{Monte Carlo $T=.00005$} &\small{Higher order $T=.005$} & \small{Monte Carlo $T=.005$}  \\ \hline
-0.10   & 20.2068\% & 20.2023\% & 20.2020\%  &20.1615\%&  20.1589
\%\\ \hline
-0.08   & 20.141\% & 20.1364\% & 20.1363\% &20.0953\%& 20.0931\% \\ \hline
-0.06   & 20.0869\% & 20.0822\%  & 20.0824\%&20.0407\%  & 20.0388\% \\ \hline
-0.04   & 20.045\% & 20.0404\% & 20.0407\%  & 19.9986\%& 19.9968\% \\ \hline
-0.02   & 20.016\% & 20.0113\% & 20.0119\% & 19.9693\% & 19.9676\%  \\ \hline
0.00   & 20.0000\% & - & 19.9942\% &-& 19.9513\% \\ \hline
0.02   & 19.9973\% & 19.9926\% & 19.9921\%  &19.9503\%& 19.9509\% \\ \hline
0.04   & 20.0079\% & 20.0033\% & 20.0029\% & 19.9610\% &  19.9613 \%  \\ \hline
0.06   & 20.0316\% &20.0270\% & 20.0266\% & 19.9850\%& 19.9850\%  \\ \hline
0.08   & 20.068\% & 20.0634\% &  20.0629\%  & 20.0218\% & 20.0213\% \\ \hline
0.10   & 20.1166\% &20.1120\% & 20.1114\%  &20.0709\%& 20.0699\%\\ \hline
\end{tabular}
\end{center}

\sk
\sk

Table of numerical results corresponding to the right plot in Figure 3 and the left plot in Figure 4.
\sk
\sk

\section{Large-time asymptotics}

\sk

In this section, we derive large-time large deviation asymptotics for the Rough Heston model, and we begin making the following assumption throughout this section:
\begin{Assumption}
$\lm>0$, $\rho \le 0$.
\end{Assumption}

\sk

\sk
Recall that $f(p,t)$ in \eqref{eq:SP} satisfies
\be
D^{\al}f(p,t) = H(p,f(p,t)) 
\label{eq:Reel}
\ee
subject to $f(p,0)=0$, where
$
  H(p,w):=\half p^2 -\half p +  (p \rho \nu-\lm) w +\half \nu^2  w^2.
$
We write
\be
  U_1(p) := \frac{1}{\nu^2} [\lm - p\rho \nu -\sqrt{\lm^2-2\lm \rho \nu p +\nu^2 p (1-p\bar{\rho}^2 })] \nn
\ee
for the smallest root of $H(p,.)$, and note that $U_1(p)$ is real if and only if $p \in [\ul{p},\bar{p}]$, where
\begin{align}
\ul{p}:=\frac{\nu-2\lm \rho-\sqrt{4\lm^2+\nu^2-4\lm \rho \nu}}{2\nu (1-\rho^2)} \quad , \quad \bar{p}:=\frac{\nu-2\lm \rho+\sqrt{4\lm^2+\nu^2-4\lm \rho \nu}}{2\nu (1-\rho^2)} \nn \,.
\end{align}

\begin{prop}
\label{prop:props}
\be
V(p):=\lim_{t\to \infty}\frac{1}{t}\log \Ex(e^{p X_t}) =   \nn
\begin{cases}
  \lm \theta U_1(p) \quad \quad  \,\,p \in [\ul{p},\bar{p}], \\
  +\infty  \quad  \quad \quad \quad p\notin [\ul{p},\bar{p}].
\end{cases}
\ee
\end{prop}

\sk

\nind \begin{proof}
\sk
\cite{GGP19} show that the explosion time for the Rough Heston model $T^*(p)<\infty$ if and only if $T^*(p)<\infty$ for the corresponding standard Heston model (i.e. the case $\al=1$).

\sk
\sk
From the usual quadratic solution formula $\frac{-b\pm\sqrt{b^2-4ac}}{2a}$, we know that $H(p,.)$ has two distinct
real roots (or a single root) if and only if
\be
(\lm-\rho p \nu)^2 \ge (p^2-p)\nu^2\label{eq:AP}
\ee
which is the same as the condition $e_1(p)\ge 0$ in condition C) in  \cite{GGP19}.
We note that $\bar{p}$, $\ul{p}$ are the zeros of $e_1(p)$.

\bs
We now have to verify that under our assumptions that $\lm>0$ and $\rho\le 0$, $T^*(p)<\infty$
if and only $e_1(p)< 0$.
We have two cases to consider to verify this claim:
\begin{itemize}
\item Suppose $e_1(p)\ge 0$. Then case B in \cite{GGP19} is impossible by definition, and
$p\in[\ul{p},\bar{p}]$, and Eq (3.5) in \cite{FJ11} is satisfied. Eq (3.4) in \cite{FJ11} is
\[
\lm  > \rho \nu p \nn
\]
in our current notation, and by the assertion on p.769 in \cite{FJ11} that ``(3.4) is implied by (3.5)'',
we see that it holds, which is equivalent to $e_0(p)<0$. Therefore, case A is impossible.
 So we are in the non-explosive
cases C or D of the \cite{GGP19} classification.
We note that case C is by definition equivalent now to $c_1(p)>0$.
\item Suppose $e_1(p)<0$. By definition we are not in case C. And we have $p \notin
[\ul{p}, \bar{p}]$, but from p.769 in \cite{FJ11}, we know the interval $[0,1]$ is strictly contained
in $[\ul{p}, \bar{p}]$.  Hence, case D is also impossible, and we are in the explosive cases
A or B.
\end{itemize}


\nind Hence our claim is verified. We can now re-write \eqref{eq:Reel} in integral form as
\be
f(p,t) = \frac{1}{\Gamma(\al)}\int_0^t (t-s)^{\al-1}H(p,f(p,s)) ds \nn.
\ee
Clearly, we have $H(p,w)\searrow 0$ as $w\nearrow U_1(p)$.  Assume to begin with that $U_1(p)>0$ (by
an easy calculation, this is exactly case C in the \cite{GGP19} classification).  Then from the proof of Proposition 4 in \cite{GGP19}, we know that $0\le f(p,t)\le U_1(p)$.

 Moreover, $w^*=U_1(p)$ is the smallest root of $H(p,w)$, so $H(p,w)\ge H_{\delta}:=H(p,U_1(p)-\delta)$ for $w\leq  U_1(p)-\delta$ and  $\delta \in (0,U_1(p))$; hence we must have
\be
\frac{H_{\delta}}{\Gamma(\al)} \int_0^t (t-s)^{\al-1} 1_{f(p,s)\le U_1(p)-\delta} \,ds < U_1(p)\nn
\ee
for all $t>0$.  This implies that
$
\frac{H_{\delta}}{\Gamma(\al)} \, (t-1)^{\al-1} \int_1^t  1_{f(p,s)\le U_1(p)-\delta} ds < U_1(p) \nn
$, or equivalently
\be
t-1- \int_1^t  1_{f(p,s)> U_1(p)-\delta}\, ds  \le \frac{\Gamma(\al)}{H_{\delta}}U_1(p)(t-1)^{1-\al} \nn \,.
\ee

Then we see that
\begin{align}
\frac{1}{t}\int_0^t f(p,s) ds \ge \frac{1}{t}\int_1^t f(p,s) ds &\ge   \frac{1}{t}\int_1^t f(p,s) 1_{f(p,s)>U_1(p)-\delta}ds \nn \\
&\ge \frac{1}{t} (U_1(p)-\delta) (t-1- \frac{\Gamma(\al)}{H_{\delta}}U_1(p)(t-1)^{1-\al} ) \nn \\
&\ge U_1(p)-2\delta \nn
\end{align}
for $t$ sufficiently large.  Thus
$
U_1(p)-2\delta \le \frac{1}{t}\int_0^t f(p,s) ds \le U_1(p) \nn \,
$, so $\frac{1}{t}\int_0^t f(p,s) ds \to U_1(p)$ as $t\to \infty$.  Then using that
\be
\log \Ex(e^{p X_t}) = V_0 I^{1-\al}f(p,t) +\lm \theta I f(p,t) \nn \,
\ee
and that $f(p,t)$ is bounded, the result follows.   We proceed similarly for the case $U_1(p)<0$ (i.e. case D in the \cite{GGP19} classification, see also Lemma \ref{lem:Benlemma}).\end{proof}

\sk
\sk
\begin{cor}
 $X_t/t$ satisfies the LDP as $t \to \infty$ with speed $t$ and rate function $V^*(x)$ equal to the Fenchel-Legendre transform of $V(p)$, as for the standard Heston model.
\end{cor}
\begin{proof}
Since $U_1'(p)\to +\infty$ as $p\to \bar{p}$ and $U_1'(p)\to -\infty$ as $p\to \ul{p}$, the function
 $\lm \theta U_1(p)$ is essentially smooth; so the stated LDP follows from the G\"{a}rtner-Ellis theorem in large deviations theory.
\end{proof}
\begin{rem}
We can easily add stochastic interest rates into this model by modelling the short rate $r_t$ by
an independent Rough Heston process, and proceeding as in \cite{FK16} (we omit the details),
see also \cite{F11}.
\end{rem}

\sk


\sk

Note that we have not proved that $f(p,t)\to U_1(p)$, but to establish the leading order behaviour in
Proposition \ref{prop:props}, this is not necessary, rather we only needed to show that $I^1 f(p,t)\sim t U_1(p)$.  Nevertheless, this convergence would be required to go to higher order, so for completeness we prove this property as well, as a special case of the following general result:

\begin{lemma}
\label{lem:Benlemma}
Consider functions $G(y)$ and $K(z)$ which satisfy the following:
\begin{itemize}
    \item $G(y)$ is analytic and increasing on $[0,y_0]$ and decreasing on $[y_0,\infty)$ where $y_0\geq 0$;
    \item $G(0)\geq 0$;
    \item $K(z)$ is positive, continuous and strictly decreasing for $z>0$;
    \item $\int_0^tK(z)dz$ is finite for each $t>0$ and diverges as $t\rightarrow \infty$;
    \item $K(z+\al)/K(z)$ is strictly increasing in z for each fixed $\al$ greater than zero.
\end{itemize}
Then the solution to $y(t)=\int_0^t K(t-s)G(y(s))ds$ is monotonically increasing, and if $G$ has at least one positive root then $y(t)$ converges to the smallest positive root of $G$ as $t\rightarrow \infty$.
\end{lemma}
\begin{proof}
See Appendix C.
\end{proof}

\bs

This lemma can be applied to both cases C and D.  As shown in \cite{GGP19}, the solution in case C is bounded between zero and the smallest positive root of $G$ (denoted $a$ in that paper) so $G$ need only satisfy the conditions of the above lemma on the interval $[0,a]$ which it does with $y_0=0$.  For case D, multiplying the defining integral equation by $-1$ and applying the transformations $-y(t)\rightarrow y(t)$ and $-G(-y(t))\rightarrow G(y(t))$ (see final plot in Figure 3) we recover an integral equation of the desired form (again $G$ need only satisfy the conditions of the lemma over the corresponding interval $[0,a]$).

\sk


\sk

\subsection{Asymptotics for call options and implied volatility}

\sk
\sk

\begin{cor}
We have the following large-time asymptotic behaviour for European put/call options in the large-time, large log-moneyness regime:
\begin{align}
\label{eq:caseI}
-\lim_{t \rightarrow \infty}\frac{1}{t}\log \mathbb{E}(S_{t}-S_0 e^{x t})^{+}&=V^*(x)-x \,\,\,\,\,\,\,\,\,\,\,\,(x \ge \half\bar{\theta})\,, \nn\\
-\lim_{t \rightarrow \infty}\frac{1}{t}\log (S_0-\mathbb{E}(S_{t}-S_0 e^{xt})^{+})&=V^*(x)-x \,\,\,\,\,\,\,\,\,\,\,\,\,(-\half\theta \le x\le\half\bar{\theta})\,, \nn\\
-\lim_{t \rightarrow \infty}\frac{1}{t}\log (\mathbb{E}(S_0 e^{xt}-S_t)^{+})&=V^*(x)-x \,\,\,\,\,\,\,\,\,\,\,\,\,(x\le -\half\theta)\,,\nn
\end{align}
where $\bar{\theta}=\frac{\lm \theta}{\lm-\rho\nu}$.
\end{cor}
\sk
\begin{proof}
See Corollary 2.4 in \cite{FJ11}.
\end{proof}

\sk

\label{cor:ImpliedVol}
\begin{cor}
We have the following asymptotic behaviour in the large-time, large log-moneyness regime, where $\hat{\sigma}_t(kt)$ is the implied volatility of a European put/call option with strike $S_0 e^{xt}$:
\be
\hat{\sigma}_{\infty}(x)^2 = \lim_{t \rightarrow \infty}\hat{\sigma}_t^2(x t)=
\frac{\omega_1}{2}(1+\omega_2 \rho x+\sqrt{(\omega_2 x+\rho)^2+\bar{\rho^2}})\nn
\ee
where
\be
\omega_1= \frac{4\lm \theta}{\nu^2 \bar{\rho}^2 }[\sqrt{(2\lm-\rho \nu)^2+\nu^2 \bar{\rho}^2}-(2\lm-\rho \nu)]\quad ,\quad
\omega_2 =\frac{\nu}{\lm \theta} \,\nn \,.
\ee
\end{cor}
\begin{proof}
See Proposition 1 in \cite{GJ11} (note that for the Rough Heston model $\lm$ has to be replaced with
$\frac{\lm}{\Gm(\al)}$ and $\nu$ replaced with
$\frac{\nu }{\Gm(\al)}$, but the effect of the $\al$ here cancels out in the final formula for $\hat{\sigma}_{\infty}(k)$.
\end{proof}

\subsection{Higher order large-time behaviour}

\sk
We can formally try going to higher order; indeed, using the ansatz $f(p,t)=U_1(p)t +U_2(p) t^{-\al}(1+o(1))$ for $p\in [\ul{p},\bar{p}]$, and we find that
\be
U_2(p) = -\frac{U_1(p)}{(\lm - U_1(p) \nu^2 - p \rho \nu) \Gamma(1 - \al)} \nn
\ee
but if we try and go higher order again, the fractional derivative on the left hand side of \eqref{eq:FDE0} does not exist.  Using the same approach as in \cite{FJM11}, one should be able to use this to compute a higher order large-time saddlepoint approximation for call options.  For the sake of brevity, we defer the details of this for future work.


\section{Asymptotics in the $H\to0$ limit}
\sk

In this section, we will show that for fixed $t$, the log stock price $X_t^{(\alpha)}:=X_t$
converges as $\alpha\to\half$ i.e. as $H\to0$ in an appropriate sense.
To match the assumptions of Theorem 13.1.1 on p.384 of \cite{GLS90} (on the continuity of the solutions to a parametrized
family of VIEs),
we define $h(\al,w) := G(p,w)$ for $\al \ge \half$ (which is independent of $\al$).
The kernel
$a(t,s,\al):=(t-s)^{\alpha-1}/\Gamma(\alpha)$
is of continuous type; see Definition 9.5.2 in \cite{GLS90}, and the remark to Theorem 12.1.1 in [GLS90], which states local integrability of $k$ as a sufficient condition for this
property, and we can easily verify that
\bq
\sup_{t \in [0,T]} |\int_0^t (a(t,s,\al)-a(t,s,\half))ds| \,\,\to\,\, 0 \nn \,
\eq
as $\al\to \half$, so the uniform
continuity assumption in Theorem 13.1.1 of \cite{GLS90} is satisfied. Moreover the solution to the VIE is unique
for $\al \in (0,1)$, see Theorem 3.1.4 in \cite{Brun17}, or Satz 1 in \cite{Di58}. Note that the Lipschitz
condition (3.1) in \cite{Di58} has a fixed Lipschitz constant $\Gamma(\alpha+1)$, but since the function $H$
defining our VIE (see \eqref{eq:Reel}) does not depend on time, the factor $t^{\alpha}$
on the left hand side of condition (3.1) in \cite{Di58} (using our notation) allows
for an arbitrary Lipschitz constant, on a sufficiently small time interval.
Moreover, once uniqueness on a small time
interval is established, there is a unique continuation (if any) by a standard extension
procedure described on p.107 of \cite{Brun17}.

\sk

Then from Theorem 13.1.1 ii) in \cite{GLS90}, $f(p,t;\al)$ is continuous in $\al$ and $t$ on
$\{(\al,t) : \al\in[\half,1), 0\leq t < \hat{T}_{\al}(p)\},$
where $[0,\hat{T}_{\al}(p))$ denotes the maximal interval on which a continuous solution
of the VIE exists.  Moreover, since Theorem 13.1.1 of \cite{GLS90} is multi-dimensional, we can apply it to
$(\mathrm{Re}(f),\mathrm{Im}(f))$ to conclude that
$f(i\theta,t;\al)\to f(i\theta,t;\half)$ for $\theta \in \mathbb{R}$.
Using the analyticity of $f(.,t,0)$, e.g.\ from Lemma 7 in \cite{GGP19},  we have that $f(i\theta,t;\half)$ is continuous at $\theta=0,$ so we can apply L\'{e}vy's convergence theorem
and verify that $X^{(\al)}_t$ tends weakly to some random variable $X^{(\half)}_t$ as $\al\to\half$, for which
\be
\Ex(e^{p X^{(\half)}_t})=e^{V_0 I^{\half}f(p,t)+\lm \theta I^1 f(p,t)}\nn
\ee for $p$ in some open interval $I=(p_-(t),p_+(t))\supset [0,1]$, where $f(p,t)$ satisfies
\be
D^{\half}f(p,t) = \half (p^2-p) +   ( p \rho \nu-\lm) f(p,t) +\half \nu^2  f(p,t)^2 \nn
\ee
with initial condition $f(p,0)=0$.

\sk
Thus we have a $H=0$ ``model'', or more precisely a family of marginals for $X^{(\half)}_t$ for all $t \in [0,T]$), with non-zero skewness.  This is in contrast to the Rough Bergomi model, which for the vol-of-vol $\gm \in (0,1)$ tends to a model with zero skew in the limit as $H\to 0$, see \cite{FFGS20} for details.

\sk
Then using similar scaling arguments to section 3, we know that
\be
\Ex(e^{p X^{(\half)}_{\e t}})=e^{V_0 I^{\half}f_{\e}(p,t)+\e^{\half} \lm \theta I^1 f_{\e}(p,t)}\nn
\ee for $p\in (p_-(\e t),p_+(\e t))\supset [0,1]$, where $f_{\e}(p,t)$ satisfies
\be
D^{\half}f_{\e}(p,t) = \half \e(p^2-p) +   \e^{\half}( p \rho \nu-\lm) f_{\e}(p,t) +\half \nu^2  f_{\e}(p,t)^2 \nn
\ee
with initial condition $f_{\e}(p,0)=0$.  Then setting
 $f_{\e}(\frac{p}{\sqrt{\e}},t)=\phi_{\e}(p,t)$ as in Eq 49 in \cite{FSV19}, we find that $\phi_{\e}(p,t)$ satisfies
\be
D^{\half}\phi_{\e}(p,t) = \half p^2  \,-\, \half p \sqrt{\e} \,+\,
   p\rho \nu \phi_{\e}(p,t) +
 \half  \nu^2 \phi_{\e}(p,t)^2 \,-\, \lm \e^{\half}\phi_{\e}(p,t)
\label{eq:VIEReal}\,
\ee
with $\phi_{\e}(p,0)=0$, for $p\in (\frac{p_-(\e t)}{\sqrt{\e}},\frac{p_+(\e t)}{\sqrt{\e}})$.  We can then apply Theorem 13.1.1 in \cite{GLS90} as above to
show that $\phi_{\e}(p,t)$ tends to the solution $\phi$ of
\bq
D^{\half}\phi(p,t) =
\half  p^2+   p \rho \nu   \phi(p,t) +\half \nu^2  \phi(p,t)^2 \label{eq:VIEEEE}
\eq
as $\e\to 0$ for $p \in (p^0_-,p^0_+)$ where $p^0_{\pm}:=\lim_{\e\to 0}\frac{p_{\pm}(\e t)}{\sqrt{\e}} $.  Thus setting $t=1$, we see (again using L\'{e}vy's convergence theorem) that $X^{(\half)}_{\e }/\sqrt{\e}$ tends weakly to a (non-Gaussian) random variable $Z$ as $t\to 0$ for which $\Ex(e^{p Z})=e^{V_0 I^{\half} \phi(p,.)(1)}$.  Two interesting and difficult open questions now arise: is this property \textit{time-consistent}, i.e. does it remain true at a future time $t$ when we condition on the history of $V$ up to $t$, and ii) is $V$ itself a well defined process in the $\al \to \half$ limit, or does it e.g. tend to a non-Gaussian field which is not pointwise defined. We answer the second question in subsections 5.2 and 5.3 below. 

\begin{rem}
Note that the scaling property in this case simplifies to
\be
\Lm(p,t) =  \Lm(p t^{\half},1)   \label{eq:SRHzero}\,
\ee
where $\Lambda(p,t):=I^{1-\al}\phi(p,t)$ with $\al=\half$.
\end{rem}

\subsection{Implied vol asymptotics in the $H=0$, $t\to 0$ limit - full smile effect for the Edgeworth FX options regime}

\sk
\sk

Following a similar argument to Lemma 5 in [MT16] one can establish the following small-time behaviour for European put options in the Edgeworth regime:
\bq
\frac{1}{\sqrt{t}}\Ex((e^{x\sqrt{t}}-e^{X_t } )^+)
\,\,\,\,\sim \,\,\,\, e^{x \sqrt{t}}\, \Ex((x-\frac{X_t}{\sqrt{t}})^+)
\,\,\,\,\sim \,\,\,\,  \Ex((x-\frac{X_t}{\sqrt{t}})^+)  &\sim& \, P(x):=\Ex((x-Z)^+) \nn \,
\eq
as $t\to 0$, where $Z$ is the non-Gaussian random variable defined in the previous subsection, and $f\sim g$ here means that $f/g\to 1$.  From e.g. [Fuk17] or Lemma 3.3 in \cite{FSV19}, we know that for the Black-Scholes model
with volatility $\sigma$
\bq
\frac{1}{\sqrt{t}}\Ex((e^{x\sqrt{t}}-e^{X_t } )^+)
&\sim& P_B(x,\sigma) \,\,\,\,:= \,\,\,\,\Ex((x- \sigma W_1)^+) \,
\eq
where $W$ is a standard Brownian motion.  From this we can easily deduce that
\bq
\hat{\sigma}_0(x) :=\lim_{t \to 0}\hat{\sigma}_t(x \sqrt{t},t) = P_B(x,.)^{-1}(P(x))\label{eq:QF} \,
\eq
for $x>0$, where $\hat{\sigma}_t(x,t)$ denotes the implied volatility of a European put option with strike $e^{x}$, maturity $t$ and $S_0=1$, and $P_B(x,\sigma)$ is the Bachelier model put price formula.  Hence we see the full smile effect in the small-time FX options Edgeworth regime unlike the $H>0$ case discussed in e.g. \cite{Fuk17}, \cite{EFGR19}, \cite{FSV19}, where the leading order term is just Black-Scholes, followed by a next order skew term, followed by an even higher order term.

\subsection{A closed-form expression for the skewness, the $H\to 0$ limit and calibrating a time-dependent correlation function}

\sk
We now consider a driftless version of the model where
$dX_t = \sqrt{V_t} dB_t$ and
$V_t = V_0 +\frac{1}{\Gamma(\al)} \int_0^t (t-s)^{\al-1}  \nu \sqrt{V_s}dW_s$.  Then
\bq
\Ex(X_T^3) ~ 3\Ex(X_T \la X\ra_T) &=& 3 \Ex(\int_0^T \sqrt{V}_s (\rho dW_s+\bar{\rho}dB_s) \int_0^T V_t dt)  ~ 3\rho\,\Ex(\int_0^T \sqrt{V}_s dW_s \int_0^T V_t dt) \nn
\eq
so formally we need to compute
\begin{align}
\Ex(\sqrt{V}_s V_t dW_s) &= \Ex(\sqrt{V}_s (V_0+
\frac{1}{\Gamma(\al)} \int_0^t (t-u)^{\al-1}  \nu \sqrt{V_u}dW_u) dW_s)\nn \\
&= \Ex(\sqrt{V}_s \frac{\nu}{\Gamma(\al)} (t-s)^{\al-1}   \sqrt{V_s}ds \,1_{s<t})\nn \\
&= \frac{\nu}{\Gamma(\al)} (t-s)^{\al-1}   \,1_{s<t} \,\Ex(V_s) ds~  \frac{\nu}{\Gamma(\al)} (t-s)^{\al-1}   \,1_{s<t} \,V_0 ds\nn \,.
\end{align}
Thus
\bq
\Ex(X_T^3) &=&  3\rho \int_0^T \int_0^t \Ex(\sqrt{V}_s V_t dW_s) ~ \frac{3V_0 \rho \nu T^{1+\al} }{\Gamma(\al)\al(1+\al)}
\label{eq:Hskew}\,.
\eq

If we now relax the assumption that $V$ is driftless and assume a given inital variance curve $\xi_0(t)$ and a general $L^2$ kernel $\kappa$ then
\bq
V_t &=& \xi_0(t) \,+\,\int_0^t \kappa(t-s)\sqrt{V}_s dW_s\nn\,
\eq
(where $\kappa$ is computed in Proposition \ref{prop:Mittag}).
Then
\bq
\Ex(\sqrt{V}_s V_t dW_s) &=& \Ex(\sqrt{V}_s (\xi_0(t)+
 \int_0^t \kappa(t-u) \sqrt{V_u}dW_u) dW_s)~ \kappa(t-s) \,1_{s<t}\Ex(V_s)\nn \,
\eq
and
\bq
\Ex(X_T^3) &=&  3\rho \int_0^T \int_0^t \Ex(\sqrt{V}_s V_t dW_s) ~
3\rho \int_0^T \int_0^t\kappa(t-s) \,\xi_0(s) ds dt\nn\,.
\eq

\begin{rem}
If we allow $\rho$ to be time-dependent, then $\Ex(X_t^3) =3\rho(t) \int_0^T \int_0^t\kappa(t-s) \,\xi_0(s) ds dt$ and we can use this equation to calibrate $\rho(t)$ to the observed \textit{skewness term structure}, i.e.
the value of $\Ex(X_t^3)$ at each $t$ in some interval $[0,T]$ implied by European option prices via the Breeden-Litzenberger formula.  Note we have ignored the drift terms of $X$ to simplify the computations here but in the small -time limit these drift terms will be higher order.
\end{rem}

\subsection{Weak convergence of the $V$ process on pathspace to a tempered distribution, and the hyper-rough Heston model}

\sk
\sk
From Theorem 4.3 in \cite{JLP19} with $\al\in(\half,1)$,
$a(v)=\nu^2 v $, $\sigma(v)=\nu \sqrt{v}$, $b(v)=\lm(\theta-v)$, $A(v)=\nu^2 v$ and $f \in L^1([0,T])$, we know that
\bq
\Ex(e^{\int_0^T f(T-t) V_t dt }) =e^{V_0\int_0^T f(t) dt \,+\,\half \nu^2 V_0 \int_0^T \psi_{\al}(t)^2 dt} \nn\,
\eq
where $\psi_{\al}$ satisfies the Riccati-Volterra equation:
\bq
\psi_{\al}(t) =
\int_0^t c_{\al}(t-s)^{\al-1} (f(s)+ \half \nu^2 \psi_{\al}(s)^2 )  ds\label{eq:SameVIE}
\eq
and $c_{\al}=\frac{1}{\Gamma(\al)}$.


\begin{prop}
\label{prop:Vhalf}
$V$ tends to a random tempered distribution $V^{(\half)}$ in distribution as $\al \to \half$ with respect to the strong and weak topologies (see page 2 in \cite{BDW17} for definitions), where $V^{(\half)}$ is a random tempered distribution\footnote{see e.g. \cite{DRSV17} for more details on tempered distributions} and for all $f$ in the Schwartz space $\mc{S}$ we have
\bq
\Ex(e^{\int_0^T f(T-t) V^{(\half)}_t dt }) =e^{V_0\int_0^T f(t) dt \,+\,\half \nu^2 V_0 \int_0^T \psi(t)^2 dt} \nn\,
\eq
where $\psi$ satisfies the following VIE:
\bq
\psi(t) =
\int_0^t c_{\half}(t-s)^{-\half} (f(s)+ \half \nu^2 \psi(s)^2 )  ds\,.\nn
\eq
\end{prop}
\nind \begin{proof}
See Appendix D.
\end{proof}

\bs

Let $A_t$ satisfy
$
A_t =V_0 t+ \frac{\nu}{\Gamma(\half)} \int_0^t (t-s)^{-\half}W_{A_s}ds \nn \,
$.  Then $A_t$ is of the same form as $X_t$ in \cite{Jab19}, with their $dG_0(t)=V_0 dt$.  Then from Theorem 2.5 in \cite{Jab19} (with $a=b=0$ and $c=\nu^2$) we know that
\begin{align}
\Ex(e^{\int_0^T f(T-t) dA_t }) = e^{\int_0^T F(T-s,\psi(T-s))dG_0(s) }
&= e^{V_0\int_0^T (f(T-s)+\half \nu^2\psi(T-s)^2) ds}\nn \\
&=e^{V_0\int_0^T (f(s)+\half \nu^2\psi(s)^2) ds}
\end{align}
\nind  
where $F(s,u)=f(u)+\half c u^2$,
and $\psi$ satisfies
\bq
\psi(t) =\int_0^t K(t-s) F(s,\psi(s))ds =
\int_0^t c_{\half}(t-s)^{-\half} (f(s)+ \half \nu^2 \psi(s)^2 )  ds\,\nn
\eq


The process $A_t$ here is the driftless \textit{hyper-rough} Heston model for $H=0$ discussed in the next subsection, and e note that $\psi$ satisfies the same VIE as \eqref{eq:SameVIE} (and by e.g. Theorem 3.1.4 in \cite{Brun17} we know the solution is unique), so the limiting field $V^{(\half)}$ has the same law as the random measure $dA_t$.
Moreover, from Proposition 4.6 in \cite{JR18} (which uses the law of the iterated logarithm for $B$) $A$ is a.s. not continuously differentiable but is only known to be $2\al-\e$ H\"{o}lder continuous for all $\e>0$.  Hence $A$ exhibits (non-Gaussian) ``field''-type behaviour.

\sk
\subsection{The hyper-rough Heston model for $H=0$ - driftless and general cases}

\sk
\sk
If $\lm=0$ and $\al \in (\half,1)$ and we set $A_t:=\int_0^t V_s ds$, then using the stochastic Fubini theorem, we see that
\begin{align}
A_t-V_0 t = \frac{1}{ \Gamma(\al)}\int_0^t \int_0^s (s-u)^{\al-1}  \nu \sqrt{V_u}dW_u ds
&=  \frac{1}{ \Gamma(\al)}\int_0^t \nu \sqrt{V_u}dW_u \int_u^t (s-u)^{\al-1} ds \nn \\
&= \frac{\nu}{\al \Gamma(\al)}\int_0^t (t-u)^{\al} \sqrt{V_u}dW_u  \nn \\
&\text{(using Dambis-Dubins-Schwarz time change} \nn \\
&= \frac{\nu}{\al \Gamma(\al)}\int_0^t (t-u)^{\al} dB_{A_u}  \nn \\
& \text{(where $B_t:=X_{T_t}$, $T_t=\inf \lb s: A_s >t \rb$) so $B$ is a Brownian motion}) \nn \\
&= \frac{\nu}{\al \Gamma(\al)} B_{A_u} (t-u)^{\al}|_{u=0}^t \,+\, \frac{\nu}{\Gamma(\al)} \int_0^t (t-u)^{\al-1} B_{A_u} du  \nn \\
&= \nu  I^{\al} B_{A_t} \nn\,.
\end{align}
We can now take
\bq
A_t=V_0 t+\nu  I^{\al} B_{A_t}\label{eq:Hyper}
\eq
 as the \textit{definition} of the Rough Heston model for $\al \in [\half,1)$ (i.e. allowing for the possibility that $\al=\half$), where $B$ is now a \textit{given} Brownian motion (this is the so-called \textit{hyper-rough Heston} model introduced in \cite{JR18} for the case of zero drift.  Note that for a given sample path $B_t(\omega)$, we can regard \eqref{eq:Hyper} as a (random) fractional ODE of the form:
\begin{align}
A(t) &= V_0 t \,+\,I^{\al} f(A(t))  \label{eq:MR}
\end{align}
where $f(t)=B_t(\omega)$.

\subsubsection{The case $\lm>0$}

\sk
For the case when $\lm > 0$, using \eqref{eq:Vtt} we see that
\bq
A_t-\int_0^t \xi_0(s)ds ~ \int_0^t \int_0^s \kappa(s-u) \sqrt{V_u}dW_u ds
&=&  \int_0^t  \sqrt{V_u}\int_u^t \kappa(s-u) ds\,dW_u  \nn \\
&=& \int_0^t F(t-u) \sqrt{V_u}dW_u  \text{  (where $F(t-u)=\int_u^t \kappa(s-u) ds$)} \nn \\
&=& \int_0^t F(t-u) dM_u  \nn \\
&& \text{(where $dM_t=\sqrt{V_t}dW_t$)} \nn \\
&=& \int_0^t F(t-u) d B_{A_u}  \nn \\
&& \text{(where $B_t:=M_{T_t}$, $T_t=\inf \lb s: A_s >t \rb$) so $B$ is a Brownian motion}) \nn \\
&=& B_{A_u} F(t-u)|_{u=0}^t \,+\,  \int_0^t \kappa(t-u) B_{A_u} du  \nn \\
&=& \int_0^t \kappa(t-u) B_{A_u} du  \nn
\eq
where we have used \eqref{eq:Mitt} to verify that $F(t-u)\to 0$ as $u\to t$.
\begin{figure}
\begin{center}
\includegraphics[width=155pt, height=155pt]{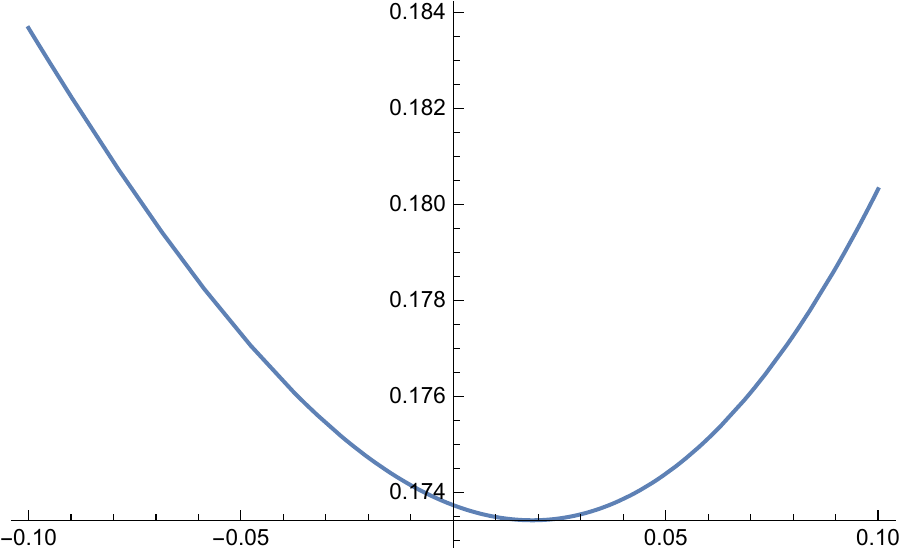}
\caption{
Here we have plotted the $H=0$ asymptotic short-maturity smile (i.e. $\hat{\sigma}_0(x)$ in \eqref{eq:QF}),
for $\nu=.2$, $\rho=-.1$ and $V_0=.04$.  We have used a 10-term small-$t$ series approximation to the solution to \eqref{eq:VIEEEE} combined with the scaling property in \eqref{eq:SRHzero}, and the Alan Lewis Fourier inversion formula for call options given in e.g. Eq 1.4 in \cite{EGR18} using Gauss-Legendre quadrature for the inverse Fourier transform with $1600$ points over a range of $[0,40]$.}
\end{center}
\end{figure}

\sk

\appendix

\renewcommand{\theequation}{A-\arabic{equation}}
\setcounter{equation}{0}
\section{Computing the kernel for the Rough Heston variance curve }
\label{section:AppA}

\sk
Let $Z_t=\int_0^t \sqrt{V}_s dW_s$, and we recall that
\begin{align}
V_t &= V_0+ \frac{1}{\Gamma(\al)} \int_0^t (t-s)^{\al-1}  \lm(\theta-V_s) ds  +\frac{1 }{\Gamma(\al)} \int_0^t (t-s)^{\al-1}  \nu \sqrt{V_s}dW_s \nn \\
&= \tilde{\xi}_0(t) -\frac{\lm}{\nu}(\ph*V) +\ph*dZ   \nn
\end{align}
where $*$ denotes the convolution of two functions,
 $\ph*dZ=\int_0^t \ph(t-s) dZ_s$ and $\tilde{\xi}_0(t)=V_0+ \frac{1}{\Gamma(\al)} \int_0^t (t-s)^{\al-1}  \lm \theta  ds=V_0 +\frac{\lm \theta}{\al \Gamma(\al)}t^{\al}$, and $\ph(t)=\frac{\nu}{\Gamma(\al)}t^{\al}$.
Now define $\kappa$ to be the unique function which satisfies
\begin{align}
\kappa  &= \ph -\frac{\lm}{\nu} (\ph * \kappa) \,.\label{eq:resolvent}
\end{align}
Such a $\kappa$ exists and is known as the \emph{resolvent} of $\ph$.  Then we see that
\begin{align}
V_t -\frac{\lm}{\nu}\,\kappa *V_t &=  \tilde{\xi}_0(t) -\frac{\lm}{\nu}\,\ph*V +\ph*dZ -
\frac{\lm}{\nu} \kappa *[\,\tilde{\xi}_0(t) -\frac{\lm}{\nu}\,\ph*V +\ph*dZ ]\nn \\
&= \xi_0(t)-\frac{\lm}{\nu}(\ph-\frac{\lm}{\nu}\,\kappa* \ph)*V+(\ph-\frac{\lm}{\nu}\kappa*\ph) *dZ \nn \\
&= \xi_0(t)-\frac{\lm}{\nu}\,\kappa*V+\kappa *dZ \nn
 \nn
\end{align}
where $\xi_0(t)=\tilde{\xi}_0(t)-\frac{\lm}{\nu}\,\kappa *\tilde{\xi}_0(t)$, and we have used \eqref{eq:resolvent} in the final line.  Cancelling the $-\frac{\lm}{\nu}\,\kappa*V$ terms, we see that
\begin{align}
V_t &= \xi_0(t)+\kappa *dZ =  \xi_0(t)+\int_0^t \kappa(t-s)\sqrt{V}_s dW_s \nn \\
\Rightarrow \quad \quad \xi_t(u) &= \Ex(V_u|\mc{F}_t) =\xi_0(u)+\int_0^t \kappa(u-s)\sqrt{V}_s dW_s \nn
\end{align}
and thus
\be
d\xi_t(u) = \kappa(u-t)\sqrt{V}_t dW_t \nn
\ee
i.e. the correct $\kappa$ function is the solution to \eqref{eq:resolvent}.  If we take the Laplace transform
of \eqref{eq:resolvent}, we get
\begin{align}
\hat{\kappa}(z)  &= \hat{\ph}(z) -\frac{\lm}{\nu} \hat{\ph}(z) \hat{\kappa}(z) \,.\label{eq:Al}
\end{align}
and \eqref{eq:Al} is just an algebraic equation now, which we can solve explicitly to get
$
\hat{\kappa}(z)=\frac{\hat{\ph}(z)}{1+\frac{\lm}{\nu} \hat{\ph}(z)}\nn\,.
$
But we know that $\ph(t)=\frac{\nu}{\Gamma(\al)}t^{\al}$ whose Laplace transform is
$\hat{\ph}(z)=\nu z^{-\al}$, so $\hat{\kappa}(z)$ evaluates to
\begin{align}
\hat{\kappa}(z) &= \frac{\nu z^{-\al}}{1+\lm z^{-\al}}\nn \,.
\end{align}
Then the inverse Laplace transform of $\hat{\kappa}(z)$ is given by
\begin{align}
\kappa(x) &= \nu x^{\al-1} E_{\al,\al}(-\lm x^{\al})\nn \,.
\end{align}

\renewcommand{\theequation}{B-\arabic{equation}}
\setcounter{equation}{0}
\section{The re-scaled model}

We first let
\label{section:AppB}
\begin{align}
dX^{\e}_t &= -\half \e V^{\e}_t dt+ \sqrt{\e} \sqrt{V^{\e}_t} dW_t  \nn \\
V^{\e}_t-V_0&=  \frac{\e^{\gm}}{\Gamma(\al)} \int_0^t (t-s)^{H-\half}  \lm(\theta-V^{\e}_s) ds  +\frac{\e^H}{\Gamma(\al)}\int_0^t (t-s)^{H-\half}  \nu \sqrt{V^{\e}_s}  dW_s \nn \\
&\eqd    \frac{\e^{\gm}}{\Gamma(\al)} \int_0^t (t-s)^{H-\half}  \lm(\theta-V^{\e}_s) ds  +\frac{\e^{H-\half}}{\Gamma(\al)}\int_0^t (t-s)^{H-\half}  \nu \sqrt{V^{\e}_s}  dW_{\e s} \nn \\
&=    \frac{\e^{\gm}}{\Gamma(\al)} \int_0^{\e t} (t-\frac{u}{\e})^{H-\half}  \lm(\theta-V^{\e}_{u/\e}) \frac{1}{\e}du  +\frac{\e^{H-\half}}{\Gamma(\al)}\int_0^{\e t} (t-\frac{u}{\e})^{H-\half}  \nu \sqrt{V^{\e}_{u/\e}}  dW_u \,.\nn
\end{align}
where we have set $u=\e s$.  Now set $V'_{\e t}=V^{\e}_t$.  Then
\begin{align}
V'_{\e t}-V_0&=\frac{\e^{\gm-1}}{\Gamma(\al)} \int_0^{\e t} (t-\frac{u}{\e})^{H-\half}  \lm(\theta-V'_u) du +\frac{\e^{H-\half}}{\Gamma(\al)}\int_0^{\e t} (t-\frac{u}{\e})^{H-\half}   \nu \sqrt{V'_{u}} \, dW_u \nn \\
 &=\frac{\e^{\gm-1}}{\e^{H-\half}\Gamma(\al)} \int_0^{\e t} (\e t-u)^{H-\half}  \lm(\theta-V'_u) du + \frac{\e^{H-\half}}{\e^{H-\half}\Gamma(\al)}\int_0^{\e t} (\e t-u)^{H-\half}  \nu \sqrt{V'_{u}} \, dW_u \nn \\
 &= \frac{1}{\Gamma(\al)}\int_0^{\e t} (\e t-u)^{H-\half} \lm(\theta-V'_u) \, du +
  \frac{1}{\Gamma(\al)}\int_0^{\e t} (\e t-u)^{H-\half}  \nu \sqrt{V'_{u}} \, dW_u\nn
\end{align}
where the last line follows on setting $\gm-1=H-\half$, i.e. $\gm=\al$.  Thus for this choice of $\gm$, $V_{\e (.)}\eqd V^{\e}_{(.)}$.


\renewcommand{\theequation}{C-\arabic{equation}}
\setcounter{equation}{0}
\section{Proof of monotonicity of the solution for a general class of Volterra integral equations}
\label{section:AppC}

\sk

\sk
Recall that $y(t)$ satisfies
\begin{equation}
    y(t)=\int_0^tK(t-s)G(y(s))ds \nn
\end{equation}
One can easily verify that the kernel used for the Rough Heston model satisfies the stated properties in Lemma \ref{lem:Benlemma}.

\sk

In the classical case $K(t)\equiv 1$ the integral eq clearly reduces to an ODE, and it is well known that the solution of this is at least continuously differentiable on the domain of existence.  In the following it will be assumed that the solution $y(t)$ is analytic for $t>0$. This is proved for the kernel relevant to the Rough Heston model in \cite{MF71} (Theorem 6), see also the end of page 14 in \cite{GGP19}.

\sk

What follows is a natural extension of the technique used in \cite{MW51} (Theorem 8).  Using the properties of convolution and differentiating under the integral sign, we have:
\begin{align}
y(t)&=\int_0^tK(t-s)G(y(s))ds=\int_0^tK(s)G(y(t-s))ds \\
y'(t)&=K(t)G(0)+\int_0^tK(s)G'(y(t-s))y'(t-s)ds \\
&= K(t)G(0)+\int_0^tK(t-s)G'(y(s))y'(s)ds
\end{align}
$G(0)>0$ so $y'(t)\rightarrow +\infty$  as $t\rightarrow 0^+$ and since $G(y)$ is increasing for $y\le y_0$ we have that $y'(t)>0$ until $y(t)$ reaches $y_0$ i.e. the solution increases.  For $y\ge y_0$, $G(y)$ is decreasing and suppose that $y(t)$ ceases to be increasing at some point. This implies (assuming a continuous derivative) the existence of a $t_0$ and an interval $I=[t_0,t_1]$ such that $y'(t_0)=0$ and $y'(t_1)<0$ for all $t_1\in I$ (if $y(t)$ and hence $y'(t)$ is analytic then the zeros of the derivative are isolated and a sufficiently small interval $I$ exists).  Using the integral equation for $y'(t)$:
\begin{align}
y'(t_0)&= K(t_0)G(0)+\int_0^{t_0}K(t_0-s)G'(y(s))y'(s)ds=0 \label{eq:annoying}\\
y'(t_1)&= K(t_1)G(0)+\int_0^{t_0}K(t_1-s)G'(y(s))y'(s)ds+\int_{t_0}^{t_1}K(t_1-s)G'(y(s))y'(s)ds \nn
\end{align}
We can re-write the kernels in the first and second terms of the expression for $y'(t_1)$ as:
\begin{align}
    K(t_1)=\frac{K(t_1)}{K(t_0)}K(t_0)\quad ,\quad    K(t_1-s)=\frac{K(t_1-s)}{K(t_0-s)}K(t_0-s)\nn
\end{align}
and we can easily check that the quotient in the second expression here decreases monotonically from $K(t_1)/K(t_0)$ to zero.\\
\\
By the mean value theorem for definite integrals there exists a $\tau\in (0,t_0)$ such that:
\begin{align}
\int_0^{t_0}\frac{K(t_1-s)}{K(t_0-s)}K(t_0-s)G'(y(s))y'(s)ds&=\frac{K(t_1-\tau)}{K(t_0-\tau)}\int_0^{t_0}K(t_0-s)G'(y(s))y'(s)ds \nn \\
&=-\frac{K(t_1-\tau)}{K(t_0-\tau)}K(t_0)G(0)
\end{align}
\nind where the second equality follows from \eqref{eq:annoying}. Substituting this into our expression for $y'(t_1)$:
\begin{align}
    y'(t_1)&= \frac{K(t_1)}{K(t_0)}K(t_0)G(0)+\frac{K(t_1-\tau)}{K(t_0-\tau)}\int_0^{t_0}K(t_0-s)G'(y(s))y'(s)ds+\int_{t_0}^{t_1}K(t_1-s)G'(y(s))y'(s)ds \nonumber\\
    &= K(t_0)G(0)\underbrace{(\frac{K(t_1)}{K(t_0)}-\frac{K(t_1-\tau)}{K(t_0-\tau)})}_{>0}+\int_{t_0}^{t_1}K(t_1-s)\underbrace{G'(y(s))y'(s)}_{>0}ds
    > 0
\end{align}
and we have used \eqref{eq:annoying} in the second line.  But this is a contradiction so the solution remains increasing.\\ \\
As discussed elsewhere in this paper, when studying the Rough Heston model, the non-linearity in the integral equation has the generic form $G(y)=(y-\theta_1)^2+\theta_2$ i.e. a quadratic with positive leading coefficient (for simplicity set to 1 here) and minimum of $\theta_2$ obtained at $y=\theta_1$. Depending on the values of $\{\theta_1,\theta_2\}$ the following cases due to \cite{GGP19} are distinguished:
\begin{itemize}
    \item (C) $G(0)>0$, $\theta_1>0$ and $\theta_2<0$
    \item (D) $G(0)\leq 0$
\end{itemize}

Case C is already in the form considered here with $y_0=0$. In case D,
applying the transformation $y(t)\rightarrow -y(t)$ and $-G(-y(t))\rightarrow G(y(t))$ (reflecting in the $x$ and then $y$ axis) yields a function $G(y)$ which is a quadratic with negative leading coefficient and thus increases until it reaches it's maximum after which it decreases which is of the type considered here.


\renewcommand{\theequation}{D-\arabic{equation}}
\setcounter{equation}{0}
\section{Appendix D}
\label{section:AppD}

From Theorem 13.1.1 ii) in \cite{GLS90}, the unique solution $\psi^{(\al)}$ to
\bq
\psi^{(\al)}(t) =
\int_0^t c_{\al}(t-s)^{\al-1} (f(s)+ \half \nu^2 \psi^{(\al)}(s)^2 )  ds\,\nn
\eq
tends pointwise to the solution of
\bq
\psi_{\half}(t) =
\int_0^t c_{\half}(t-s)^{-\half} (f(s)+ \half \nu^2 \psi_{\half}(s)^2 )  ds\,.\nn
\eq
which is also unique by e.g. Theorem 3.1.4 in \cite{Brun17}.
Now consider any sequence $f_{\e} \in \mc{S}$ with $\|f_{\e} \|_{m,j}\to 0$ as $\e\to 0$ for all $m,j \in \mathbb{N}_0^n$ for any $n \in \mathbb{N}$ (i.e. under the Schwartz space semi-norm defined in Eq 1 in \cite{BDW17}).  Then the convergence here implies in particular that $f_{\e}$ tends to $f$ pointwise.   Then from Theorem 13.1.1. in \cite{GLS90}, the unique solution $\psi_{\e}$ to
\bq
\psi_{\e}(t) =
\int_0^t c_{\half}(t-s)^{-\half} (f_{\e}(s)+ \half \nu^2 \psi_{\e}(s)^2 )  ds\,\nn
\eq
tends pointwise to the solution to
\bq
\psi_{0}(t) =
\int_0^t c_{\half}(t-s)^{-\half} \half \nu^2 \psi_{0}(s)^2 ds\,\nn
\eq
which is zero.
Then from L\'{e}vy's continuity theorem for generalized
random fields in the space of tempered distributions (see Theorem 2.3 and Corollary 2.4 in \cite{BDW17}),
we obtain the stated result.

\end{document}